\documentclass{eptcs}
\providecommand{\event}{QPL 2012} 
\usepackage{breakurl}             
\usepackage{mya4}
\usepackage{amssymb}
\usepackage{amsmath}
\usepackage{stmaryrd}
\usepackage{prooftree}
\usepackage{paralist}
\usepackage{url}
\usepackage{datetime}
\usepackage{tikz}
\usetikzlibrary{trees,arrows,positioning,calc,decorations.pathmorphing}
\usepackage[all]{xy}

\newtheorem{proposition}{Proposition}
\newtheorem{Theorem}{Theorem}
\newtheorem{Definition}{Definition}
\newtheorem{Lemma}{Lemma}
\newenvironment{proof}{\textit{Proof:}}{}
\newcommand{\qed}{\hfill$\Box$}
\newcommand{\mktype}[1]{\mathsf{#1}}
\newcommand{\Int}{\mktype{Int}}

\newcommand{\Qbit}{\mktype{Qbit}}

\newcommand{\NS}{\mktype{NS}}
\newcommand{\Bit}{\mktype{Bit}}

\newcommand{\Op}[1]{\mktype{Op}(#1)}

\newcommand{\mkterm}[1]{\mathsf{#1}}

\newcommand{\unit}{\mkterm{unit}}
\newcommand{\ctxt}[2]{{#1}[#2]}
\newcommand{\bnf}{::=}
\newcommand{\alt}{~|~}

\newcommand{\Prob}[1]{\boxplus_{#1}}
\newcommand{\ms}[1]{|{#1}|^{2}}

\newcommand{\newcnfig}[4]{({#1};{#2};{#3};{#4})}

\newcommand{\mkRrule}[1]{\mbox{\textsc{R-#1}}}

\newcommand{\Rplus}{\mkRrule{Plus}}

\newcommand{\Rmeasureq}{\mkRrule{Measure-QBIT}}

\newcommand{\Rnsmeasuret}{\mkRrule{Measure-NS-2}}

\newcommand{\Rpsmeasure}{\mkRrule{PS-Measure}}
\newcommand{\Rcontext}{\mkRrule{Context}}

\newcommand{\Rnstrans}{\mkRrule{Trans-NS}}
\newcommand{\Rifthen}{\mkRrule{IfThen}}
\newcommand{\Lps}{\mkLrule{PS}}

\newcommand{\typed}[2]{{#1}\mathrel{\!:\!}{#2}}
\newcommand{\ptyped}[2]{{#1}\vdash{#2}}

\newcommand{\ket}[1]{|#1\rangle}
\newcommand{\nil}{\mathbf{0}}
\renewcommand{\parallel}{\mathbin{\mid}}
\newcommand{\inp}[2]{{#1}?{[#2]}}
\newcommand{\outp}[2]{{#1}!{[#2]}}

\renewcommand{\vec}[1]{\widetilde{#1}}
\newcommand{\new}{\mathsf{new}\ }
\newcommand{\qbit}{\mathsf{qbit}\ }
\newcommand{\ns}{\mathsf{ns}\ }
\newcommand{\bit}{\mathsf{bit}}

\newcommand{\chant}[1]{\widehat{~}[#1]}
\newcommand{\tid}[2]{{#1}\mathrel{\!:\!}{#2}}

\newcommand{\qgate}[1]{\mathsf{#1}}
\newcommand{\pname}[1]{\mathit{#1}}
\newcommand{\action}[1]{\{#1\}}

\newcommand{\trans}{\mathbin{*\!\!=}}
\newcommand{\sep}{\,.\,}

\newcommand{\measure}{\mathsf{measure}\ }

\newcommand{\psmeasure}{\mathsf{psmeasure}\ }

\newcommand{\matr}[4]{\begin{pmatrix}{#1}&{#2}\\{#3}&{#4}\end{pmatrix}}

\newcommand{\qstore}[2]{[#1 \mapsto #2]}

\newcommand{\transition}[1]{\stackrel{#1}{\longrightarrow}}
\newcommand{\transitione}{\longrightarrow_e}
\newcommand{\transitionv}{\longrightarrow_v}
\newcommand{\transitionp}[1]{\stackrel{#1}{\longrightarrow_p}}
\newcommand{\ptrans}[1]{\stackrel{#1}\rightsquigarrow}
\newcommand{\weaktrans}[1]{\stackrel{#1}{\Longrightarrow}}
\newcommand{\opttrans}[1]{\stackrel{#1}{\longrightarrow}^{+}}

\newcommand{\cfgdistsum}{\oplus}

\newcommand{\Dist}[2]{\cfgdistsum_{#1}~#2~}

\newtheorem{example}{Example}

\newcommand{\subst}[2]{\{{#1}/{#2}\}}

\newcommand{\bra}[1]{\langle#1|}

\newcommand{\ketbra}[2]{\ket{#1}\bra{#2}}

\newcommand{\parcomp}{~\Vert~}
\newcommand{\rhoe}{\rho_E}

\newcommand{\ltrm}[3]{\lambda{#1}\bullet{#2}; {#3}}
\newcommand{\ltrmshort}[2]{\lambda{#1}\bullet{#2}}
\newcommand{\states}{\ensuremath{\mathcal{S}}}

\newcommand{\ntates}{\ensuremath{\mathcal{T}_n}}
\newcommand{\ptates}{\ensuremath{\mathcal{T}_p}}

\newcommand{\mkLrule}[1]{\mbox{\textsc{L-#1}}}
\newcommand{\mkPrule}[1]{\mbox{\textsc{P-#1}}}

\newcommand{\Lexpr}{\mkLrule{Expr}}
\newcommand{\Lcom}{\mkLrule{Com}}
\newcommand{\Lnew}{\mkLrule{Res}}
\newcommand{\Lqbit}{\mkLrule{Qbit}}
\newcommand{\Loutq}{\mkLrule{Out-Qbit}}
\newcommand{\Loutns}{\mkLrule{Out-Ns}}

\newcommand{\Lpar}{\mkLrule{Par}}
\newcommand{\Lns}{\mkLrule{Ns}}

\newcommand{\Lsum}{\mkLrule{Sum}}

\newcommand{\Lin}{\mkLrule{In}}

\newcommand{\Lprob}{\mkLrule{Prob}}
\newcommand{\Lact}{\mkLrule{Act}}

\newcommand{\Pout}{\mkPrule{Out}}
\newcommand{\Pin}{\mkPrule{In}}
\newcommand{\Ppar}{\mkPrule{Par}}
\newcommand{\Psum}{\mkPrule{Sum}}
\newcommand{\Pnew}{\mkPrule{Res}}
\newcommand{\Pps}{\mkPrule{PS}}

\newcommand{\pbsim}{\leftrightarroweq}
\newcommand{\fpbsim}{\leftrightarroweq^c}

\renewcommand{\vec}[1]{\widetilde{#1}}
\newcommand{\trace}{\mathrm{tr}}

\newcommand{\ppolse}{\pname{PolSe}}

\newcommand{\ppolsect}{\pname{PolSe_{CT}}}

\newcommand{\ppsm}{\pname{PSM}}
\newcommand{\pmmt}{\pname{MMT}}
\newcommand{\pbs}{\pname{BS}}

\newcommand{\pdet}{\pname{Det}}
\newcommand{\pdete}{\pname{PDet}}
\newcommand{\pmodel}{\pname{Model_1}}
\newcommand{\pmodelt}{\pname{Model_2}}

\newcommand{\pdetone}{\pname{Det_{1}}}
\newcommand{\pdettwo}{\pname{Det_{2}}}

\newcommand{\pcounter}{\pname{Counter}}

\newcommand{\pcnot}{\pname{CNOT}}

\newcommand{\poutp}{\pname{Output}}
\newcommand{\pop}{\pname{OP}}
\newcommand{\pspec}{\pname{Specification_1}}
\newcommand{\pspect}{\pname{Specification_2}}

\input{Qcircuit}

\title{Verification of Linear Optical Quantum Computing using Quantum Process Calculus}
\author{Sonja Franke-Arnold
\institute{School of Physics and Astronomy\\
University of Glasgow, UK}
\email{sonja.franke-arnold@glasgow.ac.uk}
\and
Simon J. Gay
\institute{School of Computing Science\\
University of Glasgow, UK}
\email{Simon.Gay@glasgow.ac.uk}
\and
Ittoop Vergheese Puthoor \thanks{Supported by a Lord Kelvin / Adam Smith Scholarship from the University of Glasgow.}
\institute{School of Computing Science and\\
School of Physics and Astronomy\\
University of Glasgow, UK}
\email{ittoop@dcs.gla.ac.uk}
}
\def\titlerunning{Verification of linear optical quantum computing using 
quantum process calculus}
\def\authorrunning{S. Franke-Arnold, S. J. Gay and  I. V. Puthoor}
\begin{document}
\maketitle


\begin{abstract}
We explain the use of quantum process calculus to describe and analyse linear optical quantum computing (LOQC). The main idea is to define two \emph{processes}, one modelling a linear optical system and the other expressing a specification, and prove that they are \emph{behaviourally equivalent}. We extend the theory of \emph{behavioural equivalence} in the process calculus Communicating Quantum Processes (CQP) to include multiple particles (namely photons) as information carriers, described by \emph{Fock states} or \emph{number states}. We summarise the theory in this paper, including the crucial result that equivalence is a \emph{congruence}, meaning that it is preserved by embedding in any context. In previous work, we have used quantum process calculus to model LOQC but without verifying models against specifications. In this paper, for the first time, we are able to carry out verification. We illustrate this approach by describing and verifying two models of an LOQC CNOT gate.

\end{abstract}

\section{Introduction}
\label{sec-intro}
\label{sec:intro}
Quantum information processing (QIP) is a field of research, which involves the study of storing and manipulating information in systems that are governed by the laws of quantum mechanics. This provides huge potential in quantum computation, cryptography and communication \cite{Nielsen2000}, and first secure cryptography systems are already commercially available \cite{IDQ2001a}. Linear optical quantum computing (LOQC) is being pioneered for applications in scalable quantum computing \cite{Knill2001}.  LOQC is based on \emph{spatial encoding} where a quantum bit is encoded as a superposition of two spatial modes or the two optical paths that can be travelled by a single photon \cite{Brien2003}. The inherent weak interaction between photons as information carriers makes them highly suitable for communication applications. 

Quantum process calculus is a class of \emph{formal methods}, able to describe and analyse the behaviour of systems that combine quantum and classical elements. The success of formal methods in classical computer science has motivated the development of quantum process calculus called Communicating Quantum Processes (CQP) \cite{Gay2005}. CQP provides an abstract model of the quantum system, with the assumption that a qubit is considered as a localised unit of information. CQP verifies the correctness of a system by employing the theory of \emph{behavioural equivalence} \cite{DavidsonThesis} between processes. Also, the equivalence is a \emph{congruence}, meaning that it is preserved by inclusion in any environment. The theory has been applied to the analysis of a quantum error correcting code  \cite{Davidson2011}. 

\paragraph*{Contributions:}

This paper enhances from previous work \cite{Arnold2013} significantly in two different ways. First, we provide the theory of equivalence in CQP for LOQC, which has been extended from Davidson thesis \cite{DavidsonThesis}, in order to analyse and verify a realistic experimental system. The \emph{congruence} property of equivalence in CQP is applied to the LOQC CNOT gate, which provides us for the first time with a more physical understanding of the property of equivalence. Second, we present two models of an experimental system that demonstrates LOQC CNOT gate and prove that they are equivalent to their specification. These two models not only demonstrates the gate but uses two different measurement semantics which exhibits the flexibility of process calculus approach to work at different levels of abstraction. In our second model, we demonstrate \emph{post-selection}, which plays an important role in LOQC, where one considers only a subset of all experimental runs that fulfil predefined criteria.

The present paper begins in Section~\ref{sec:prelim} by recalling the basic concepts of quantum optics which are needed to understand LOQC. In Section~\ref{sec:CQP} we review the language of CQP, illustrated with a model of the experimental system that demonstrates LOQC CNOT gate. Section~\ref{sec:equivalence} summarises the extension of the theory of equivalence in CQP, which is applied to LOQC. In Section~\ref{sec:model} we describe the post-selection process and analyse a model of an experimental system demonstrating \emph{post-selective} LOQC CNOT gate. Finally, Section~\ref{sec:conclusion} concludes with an indication of directions for future work. 

\paragraph*{Related Work:}

All the quantum process calculi which have been developed so far considered a qubit as an abstract particle that can be sent or received through channels. Feng \emph{et al.} \cite{Feng2011} developed qCCS, a quantum extension of the classical value-passing CCS \cite{Milner1989} and proved that weak bisimilarity is a congruence. The result is applied to quantum teleporation, superdense coding and quantum key distribution protocols \cite{Kubota2012}. 

\section{Background}
\label{sec-prelim}
\label{sec:prelim}

We recall briefly the aspects of quantum theory and quantum optics relevant for this paper. For more detailed information we refer to the book by Nielsen and Chuang \cite{Nielsen2000} and research papers \cite{Knill2001,Brien2003,Ralph2002}.

A \emph{qubit} is an information unit comprising two states ($\ket{0}$ and $\ket{1}$) which are called the \emph{standard} basis. The \emph{state space} $\mathbb{H}$ (or Hilbert space) of a qubit consists of all \emph{superpositions} of the basis states: $\ket{\psi} = \alpha\ket{0} + \beta\ket{1}$ where $\alpha$ and $\beta$ are complex numbers such that $|\alpha|^{2} + |\beta|^{2} = 1$. A qubit is conventionally realised by an individual photon with the two basis states refering to orthogonal polarisation directions of the photon ($\ket{0} = \ket{H}$ and $\ket{1} = \ket{V}$). We refer to the qubit as a polarisation qubit where $H$ and $V$ denote horizontal and vertical polarisation, respectively.  We introduce the notation $\alpha\ket{H} + \beta\ket{V} = \alpha\ket{10}_{HV} + \beta\ket{01}_{HV}$, where the entries in the ket states represent the number of photons (photon number $n$) in the state basis indicated by the subscripts. This will allow us to generalise the notation to more than one photon. Two photons in the states $\alpha_{i}\ket{H} + \beta_{i}\ket{V}$ (where $i$ is 1,2 respectively for each photon) can then be encoded in the shorthand $\alpha_{1}\alpha_{2}\ket{20}_{HV} + \beta_{1}\beta_{2}\ket{02}_{HV} + (\alpha_{1}\beta_{2} + \alpha_{2}\beta_{1})\ket{11}_{HV}$, if they are indistinguishable in all other parameters. In LOQC \cite{Knill2001}, we consider qubits which are encoded in different optical paths  '$a$' and '$b$' rather than different polarisation states. This is referred to as \emph{dual rail logic}. Again, we denote the quantum states in the \emph{number state} basis, giving the number of photons travelling along the different paths. The basis states in dual rail logic are then $\ket{0} \rightarrow \ket{10}_{ab}$, and similarly for $\ket{1} \rightarrow \ket{01}_{ab}$. 
In experiments, the conversion of a \emph{polarisation} qubit into a \emph{dual rail} qubit is accomplished by the combination of a polarising beam splitter (PBS) and a phase shifter (PR) \cite{Brien2003}, which works as a unitary operation $\mathsf{PS}$.

\begin{Definition}[$\mathsf{PS}$ operator]
 A $\mathsf{PS}$ is an operator that transforms a polarisation qubit $\ket{\psi} \in \mathbb{H}_{q}$ to a dual rail qubit $\ket{\phi} \in \mathbb{H}_{s}$, where $\mathbb{H}_{q}$ and $\mathbb{H}_{s}$ are the respective Hilbert spaces for the polarisation and dual rail qubits. The action of $\mathsf{PS}$ is then defined by
 $\mathsf{PS}\ket{H} \equiv \mathsf{PS}\ket{10}_{HV} = \ket{10}_{ab}$ and
$\mathsf{PS}\ket{V} \equiv \mathsf{PS}\ket{01}_{HV} = \ket{01}_{ab}$
\label{def:ps}
\end{Definition}

Operations on number states (or \emph{Fock states} $\ket{n}$) are described in terms of the creation and annihilation operators $\hat{a}^{\dagger}$ and $\hat{a}$, which when acting on a state $\ket{n}$ increase or decrease the photon number ($n$) by one. Therefore, each Fock state can be built up from creation operators given by $\ket{n} = \frac{(\hat{a}^{\dagger})^{n}}{\sqrt{n!}}\ket{0}$. In LOQC, optical elements such as phase shifters and non polarising beam splitters perform  \emph{unitary transformations}, which describe the evolution of a closed quantum system. A unitary transformation in LOQC \cite{Myers2005} can be described by its effect on each photon path's creation operator. A non polarising beam splitter (BS) is defined by the transformation matrix 
\[
\begin{array}{rcl}
U(BS) & = & \matr{\cos\theta}{e^{i\phi}\sin\theta}{e^{-i\phi}\sin\theta}{-\cos\theta}
\end{array}
\]

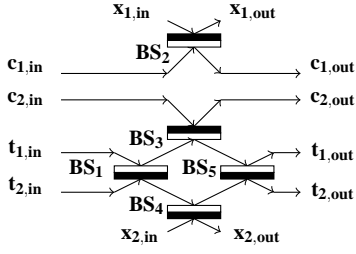
\begin{figure}
\begin{tikzpicture}[scale=0.35]
\draw [fill=black] (21,1.25) rectangle (23,1.5);\draw [fill=white] (21,1) rectangle (23,1.25);
\draw [->]  (22,1.5) -- (23,2); \node at (24.25,2.3){$\bold{_{x_{1,out}}}$};
\draw [->]  (21,2) -- (22,1.5); \node at (19.7,2.3){$\bold{_{x_{1,in}}}$};
\node at (15.7,-1){$\bold{_{c_{2,in}}}$};
\node at (15.7,0.25){$\bold{_{c_{1,in}}}$};
\node at (27.25,-1){$\bold{_{c_{2,out}}}$};
\node at (27.25,0.25){$\bold{_{c_{1,out}}}$};
\draw [->]  (17,0) -- (21,0); \draw [->]  (21,0) -- (22,1);\draw [->]  (22,1) -- (23,0);\draw [->]  (23,0) -- (26,0);
\draw [->]  (17,-1) -- (21,-1); \draw [->]  (21,-1) -- (22,-2);\draw [->]  (22,-2) -- (23,-1);\draw [->]  (23,-1) -- (26,-1);
\draw [fill=white] (21,-2.25) rectangle (23,-2.5);\draw [fill=black] (21,-2) rectangle (23,-2.25);
\draw [->]  (17,-3) -- (19,-3); \draw [->]  (19,-3) -- (20,-3.5);\draw [->]  (20,-3.5) -- (22,-2.5);
\draw [fill=black] (19,-3.75) rectangle (21,-4);\draw [fill=white] (19,-3.5) rectangle (21,-3.75);\node at (18,-3.6){$\bold{_{BS_{1}}}$};\node at (20.5,0.75){$\bold{_{BS_{2}}}$};\node at (20.25,-2.5){$\bold{_{BS_{3}}}$};\node at (20.25,-5){$\bold{_{BS_{4}}}$};\node at (22.25,-3.6){$\bold{_{BS_{5}}}$};
\draw [->]  (17,-4.5) -- (19,-4.5); \draw [->]  (19,-4.5) -- (20,-4);\draw [->]  (20,-4) -- (22,-5);
\draw [fill=black] (21,-5.25) rectangle (23,-5.5);\draw [fill=white] (21,-5) rectangle (23,-5.25);
\draw [->]  (22,-2.5) -- (24,-3.5); \draw [->]  (24,-3.5) -- (25,-3);\draw [->]  (25,-3) -- (26,-3);
\node at (15.7,-4.3){$\bold{_{t_{2,in}}}$};
\node at (15.7,-2.9){$\bold{_{t_{1,in}}}$};
\node at (27.25,-4.5){$\bold{_{t_{2,out}}}$};
\node at (27.25,-3){$\bold{_{t_{1,out}}}$};
\draw [->]  (22,-5) -- (24,-4); \draw [->]  (24,-4) -- (25,-4.5);\draw [->]  (25,-4.5) -- (26,-4.5);
\draw [fill=black] (23,-3.75) rectangle (25,-4);\draw [fill=white] (23,-3.5) rectangle (25,-3.75);
\draw [->]  (22,-5.5) -- (23,-6); \node at (24.4,-6.2) {$\bold{_{x_{2,out}}}$};
\draw [->]  (21,-6) -- (22,-5.5); \node at (20,-6.2) {$\bold{_{x_{2,in}}}$};
\end{tikzpicture}
\caption{
A schematic representation of the LOQC CNOT gate. $BS_{1}$ and $BS_{5}$ are beam splitters of reflectivity $\frac{1}{2}$ and the others are of reflectivity $\frac{1}{3}$. The dark side of the BS indicates the side from which a sign change occurs upon reflection.
}
\label{fig:cnot}
\end{figure}

The reflectivity and transmissivity of a BS are given by $ \eta = \cos^{2}\theta$ and $1-\eta=\sin^2\theta$, respectively, $\theta$ is the angle between the polarisation direction of the input photon and the crystal axis of the BS and $\phi$ is the relative phase between the light modes in the two output paths. Here we consider $\phi = 0$, which is the case for BSs in integrated circuits. If the state $\ket{mn}_{ab}$ is incident on a BS with $m$ photons along path $a$ and $n$ photons along path $b$, the transformation is:
\begin{equation}
\label{Eq1}
\begin{array}{rcl}
\ket{mn}_{ab} =  \frac{(\hat{a}^{\dagger}_{a})^{m}}{\sqrt{m!}}\frac{(\hat{a}^{\dagger}_{b})^{n}}{\sqrt{n!}}\ket{00}_{ab} 
\rightarrow\frac{1}{\sqrt{m!n!}}(\hat{a}^{\dagger}_{a} \cos\theta + \hat{a}^{\dagger}_{b} \sin\theta )^{m}( \hat{a}^{\dagger}_{a} \sin\theta - \hat{a}^{\dagger}_{b} \cos\theta )^{n}\ket{00}_{ab}
\end{array}
\end{equation}

The \emph{controlled Not} (or $\mathsf{CNOT}$) is a quantum gate that is a primary component in building a quantum computer. The operation of the gate is that it flips the second qubit (target qubit) if and only if the first qubit (control qubit) is $1$. On qubits, we have $\mathsf{CNOT}\ket{0x} = \ket{0x}$ and $\mathsf{CNOT}\ket{1x} = \ket{1y}$ where $x,y \in \{0,1\}$ and $y = x \oplus 1$ with $\oplus$ denoting addition modulo $2$. In dual rail logic, this becomes $\mathsf{CNOT}\ket{10yx} = \ket{10yx}$ and $\mathsf{CNOT}\ket{01yx} = \ket{01xy}$.

In the following we summarise the  theory and operation of the  LOQC CNOT gate \cite{Brien2003,Ralph2002}.  The BSs used in the LOQC CNOT gate \cite{Brien2003,Ralph2002} have reflectivities of  $\eta=\frac{1}{2}$ or $\frac{1}{3}$. 
The operation is specified by a control qubit, characterised by the number states $c_1$ and $c_2$, and a target qubit, characterised by $t_1$ and $t_2$, as well as two auxiliary vacuum states (absence of a qubit or photon)  $x_1$ and $x_2$, written as $\ket{c_{1}c_{2}t_{1}t_{2}}\ket{x_{1}x_{2}}$.
Consider the general input state 
\begin{equation}
\label{Eq2}
\ket{\psi}_{\rm in}\ket{00} =(\alpha\ket{1010}+\beta\ket{1001}+\gamma\ket{0110}+\delta\ket{0101})\ket{00} 
\end{equation}
The schematic representation of the LOQC CNOT gate is shown in Figure~\ref{fig:cnot}. Using the operators for each BS as discussed in Eq.~\ref{Eq1} and applying it to the input state, Eq.~\ref{Eq2} we get the output state of the CNOT gate as:
\begin{equation}
\label{Eq3}
\begin{array}{rcl}
\ket{\psi}_{\rm in}\ket{00} \rightarrow \frac{1}{3}\{(\alpha\ket{1010}+\beta\ket{1001}+\gamma\ket{0101}+\delta\ket{0110})\ket{00} +\sqrt{2}(\alpha+\beta)\ket{0100}\ket{10}\hspace{2mm}\\+ \sqrt{2}(\alpha-\beta)\ket{0000}\ket{11}+(\alpha+\beta)\ket{1100}\ket{00}+(\alpha-\beta)\ket{1000}\ket{01}+\sqrt{2}\alpha\ket{0010}\ket{10}\hspace{4mm}\\+\sqrt{2}\beta\ket{0001}\ket{10}-\sqrt{2}(\gamma+\delta)\ket{0200}\ket{00}-(\gamma-\delta)\ket{0100}\ket{01}+\sqrt{2}\gamma\ket{0020}\ket{00}\hspace{4mm}\\+(\gamma-\delta)\ket{0010}\ket{01}+(\gamma+\delta)\ket{0011}\ket{00}+(\gamma-\delta)\ket{0001}\ket{01}+\sqrt{2}\delta\ket{0002}\ket{00} \}\hspace{4mm}
\end{array}
\end{equation}

LOQC embeds qubits into the larger dual-rail space, to enable a particular physical realisation of unitary operators to be used. However, this introduces the possibility that the result of the final measurement may be outside the embedding and hence not interpretable as a computational result. \emph{Post-selection} compensates for this possibility by discarding the undesirable measurement results at the expense of introducing a non-zero probability that the overall computation fails. From these states we \emph{post-select} only those where one photon is found in the target and one in the control state, by discarding all terms apart from the first four terms in the first line of Eq.~\ref{Eq3}, giving
\begin{equation}
\label{Eq4}
\ket{\phi}_{ps} = \alpha\ket{1010}+\beta\ket{1001}+\gamma\ket{0101}+\delta\ket{0110}
\end{equation}
Successful \emph{post-selection} occurs only with a probability of one-ninth and the relationship between Eq.~\ref{Eq2} and Eq.~\ref{Eq4} is a controlled-NOT transformation.

\section{Communicating Quantum Processes (CQP)}
\label{sec:CQP}
\label{sec-CQP}
CQP \cite{Gay2005} is a quantum process calculus, which was established for formally defining the structure and behaviour of systems that comprise both quantum and classical communication and computation. The language is based on the $\pi$-calculus \cite{Milner1999} with primitives for quantum information. The general idea is that a system is considered to be made up of independent components or \emph{processes}. The \emph{processes} can communicate by sending and receiving data along \emph{channels} and these data are qubits, number states or classical values. A distinctive feature of CQP is its static type system \cite{Gay2006a}, the purpose of which is to classify classical and quantum data and also to enforce the no-cloning property of quantum information. We now present CQP including the extensions required for LOQC.

\subsection{Syntax of CQP}

\begin{figure*}[t]
\small
  \begin{eqnarray*}
    T & \bnf & \Int \alt \Qbit \alt \NS \alt \Bit \alt \chant{\tilde{T}} \alt \Op{1} \alt \Op{2} \alt \cdots \\
    v & \bnf & x \alt \mktype{0} \alt \mktype{1} \alt \cdots \alt \qgate{H} \alt \cdots \\
    e & \bnf & v \alt \measure{\tilde{e}} \alt \psmeasure{\tilde{e}} \alt {\tilde{e}}\trans{e} \alt e+e' \alt (e,e) \alt \text{if $e$ then $e$ else $e$} \alt x:\NS, y:\NS \trans\qgate{PS}{(z)}\\
    P & \bnf & \nil \alt (P | P) \alt P + P \alt \inp{e}{\tilde{x}:\tilde{T}}.P \alt  \outp{e}{\tilde{e}}.P \alt \{e\}.P \alt (\qbit x)P \alt (\ns x)P \alt (\new x:\chant{T})P
  \end{eqnarray*}
    \normalsize
  \caption{\label{fig:cqp_syntax}Syntax of CQP.}
\end{figure*}
The syntax of CQP is defined by the grammar as shown in Figure \ref{fig:cqp_syntax}. We use the notation $\tilde{e} = e_1,\ldots,e_n$, and write $|\tilde{e}|$ for the length of a tuple. The syntax consists of types $T$, values $v$, expressions $e$ (including quantum measurements and the conditional application of unitary operators $\tilde{e}\trans{e}$), and processes $P$. Values $v$ consist of variables ($x$,$y$,$z$ etc), literal values of data types (0,1,..), unitary operators such as the Hadamard operator $\mathsf{H}$. Expressions $e$ consist of values, measurements $\measure e_{1},\dots,e_{n}$, applications $e_{1},\dots,e_{n} \trans e$ of unitary operators and applications $x :\NS,y:\NS \trans\qgate{PS}{(z)}$ of $\qgate{PS}$ operator, expressions involving data operators such as $e + e'$ and a pair of values $(e,e)$. We have a new addition to the expression called \emph{post-selective} measurement $\psmeasure e_{1},\dots,e_{n}$. Processes include the nil process $\nil$, parallel composition $P|P$, inputs $\inp{e}{\tid{\vec{x}}{\vec{T}}}.P$, outputs $\outp{e}{\vec{e}}.P$, actions $\{e\}.P$ (typically a unitary operation or measurement), typed channel restriction $(\new x: \chant{\vec{T}})P$, qubit declaration $(\qbit x)P$ and number state declaration $(\ns x)P$.

In order to define the operational semantics we provide the \emph{internal syntax} in Figure \ref{fig:cqp_internal_syntax}. We assume a countably infinite set of qubit names, ranging over $q,r,\dots$, a countably infinite set of number state names $s,t,\dots$ and similarly channel names. Values are supplemented with either qubit names $q$ or number state names $s$, which are generated at run-time and substituted for the variables used in $\qbit$ and $\ns$ declarations respectively. Evaluation contexts for expressions ($E[]$) and processes ($F[]$) are used to define the operational semantics \cite{Wright1994}. Later in the paper, we also use parameterised process definitions.

\subsection{Linear Optical Elements in CQP}
\label{sec:LOCQP}

First, we define a process $\ppolse$ which provides the input to the LOQC CNOT gate by converting a polarisation qubit to a dual rail qubit.
\[
\begin{array}{rcl}
\pname{PolSe}(\typed{a}{\chant{\Qbit}},\typed{c}{\chant{\NS}},\typed{d}{\chant{\NS}})
  & = & 
\inp{a}{\typed{q_{0}}{\Qbit}}\sep\{\typed{s_{0}}{\NS},\typed{s_{1}}{\NS}\trans\qgate{PS}(q_{0})\} \sep\outp{c}{s_{0}}\sep\outp{d}{s_{1}}\sep\nil
\end{array}
\]
$\ppolse$ is parameterized by three channels, $a$,$c$ and $d$. The right hand side of the definition specifies the behaviour of the process $\ppolse$. The polarisation qubit (say $q_{0}$) is received as input through channel $a$ (whose type is $\chant{\Qbit}$) indicated as  $\inp{a}{\typed{q_{0}}{\Qbit}}$. The term $\{\typed{s_{0}}{\NS},\typed{s_{1}}{\NS}\trans\qgate{PS}(q_{0})\}$ specifies that the $\mathsf{PS}$ operation is applied to qubit $q_{0}$ thereby generating $s_{0}$ and $s_{1}$ of type number states ($\NS$). $\mathsf{PS}$ corresponds to the transformation produced by the combination of PBS and PR, introduced by Definition~\ref{def:ps}. The last two terms ($\outp{c}{s_{0}}$ and $\outp{d}{s_{1}}$) indicate that the respective values of the number states are sent through the respective output channels. The term $\nil$ simply indicates termination.
\begin{figure*}[t]
\small
  \begin{eqnarray*}
    v & \bnf & \ldots \alt q \alt s\alt c \\
    E & \bnf & [] \alt \measure{E,\tilde{e}} \alt \measure{v, E, \tilde{e}} \alt \dots \alt \measure{\tilde{v}, E} \alt E + e \alt v + E \alt  \text{if $E$ then $e$ else $e$}\\
    F & \bnf & \inp{[]}{\tilde{x}}.P \alt \outp{[]}{\tilde{e}}.P \alt \outp{v}{[].\tilde{e}}.P \alt \outp{v}{v,[],\tilde{e}}.P \alt \cdots \alt \outp{v}{\tilde{v},[]}.P \alt \{[]\}.P
  \end{eqnarray*}
    \normalsize
  \caption{\label{fig:cqp_internal_syntax}Internal syntax of CQP.}
\end{figure*}

Next, we define a non polarising beam splitter in CQP as $\pbs$, which is a primary component in the LOQC CNOT gate.
\[
\begin{array}{rcl}
\pname{BS}(\typed{e}{\chant{\NS}},\typed{f}{\chant{\NS}},\typed{h}{\chant{\NS}},\typed{i}{\chant{\NS}}, \eta)
  & = & 
  \inp{e}{\typed{s_{2}}{\NS}}\sep\inp{f}{\typed{s_{3}}{\NS}}\sep \action{s_{2},s_{3}\trans\qgate{B_{\eta}}}\sep \outp{h}{s_{2}}\sep\outp{i}{s_{3}}\sep\nil
\end{array}
\]
where $\eta$ is the reflectivity.
In a similar way, process $\pbs$ receives inputs $s_{2}$ and $s_{3}$ from $e$ and $f$. Then performs the unitary operation represented by $\action{s_{2},s_{3}\trans\qgate{B_{\eta}}}$ on the number states as defined by Eq.~\ref{Eq1}. Here $\mathsf{B_{\eta}}$ is the unitary operation represented by the matrix $U(BS)$ for $\phi = 0$. The number states are then output on $h$ and $i$. 

In this paper, we present two types of measurements. We define $\pdet$ and $\pdete$ which represent the detectors that performs measurement and $\pdete$ performs \emph{post-selective} measurement. 
\[
\begin{array}{rcl}
\pname{Det}(\typed{l}{\chant{\NS}},\typed{m}{\chant{\NS}},\typed{u}{\chant{\Int,\Int}})
  & = & 
  \inp{l}{\typed{s_{0}}{\NS}}\sep\inp{m}{\typed{s_{1}}{\NS}}\sep\outp{u}{\measure s_{0},s_{1}}\sep\nil\\
\pname{PDet}(\typed{l}{\chant{\NS}},\typed{m}{\chant{\NS}},\typed{u}{\chant{\Bit}})
  & = & 
  \inp{l}{\typed{s_{0}}{\NS}}\sep\inp{m}{\typed{s_{1}}{\NS}}\sep\outp{u}{\psmeasure s_{0},s_{1}}\sep\nil
\end{array}
\]
Here, the detectors measure a pair of number states. The expression $\measure s_{0},s_{1}$ probabilistically evaluates to a pair of positive integers which is the number of photons detected in the respective channels and $\psmeasure s_{0},s_{1}$ produces a zero or one which is a result of \emph{post-selection}. The different measurement semantics enables us to work at different levels of abstraction by showing the flexibility of the process calculus approach and is discussed in detail in later sections of the paper.

\subsection{The LOQC CNOT Gate in CQP : First Model}

The structure of the first model of the experimental system that demonstrates LOQC CNOT gate is shown in Figure~\ref{fig:mod}. The system receives two polarisation qubits (control and target) as inputs through the channels $a$ and $b$. The qubits are then converted to number states by the process $\ppolsect$, and these are provided as the input to the CNOT gate represented by process $\pcnot$. The output of $\pcnot$ is then measured by the process $\pmmt$. The whole model is then defined as a parallel composition of $\ppolsect\parallel\pcnot\parallel\pmmt$. The CQP definition of the model is
\[
\begin{array}{rcl}
\pname{Model_1}(\vec{X})  
=  (\new \vec{Y})(\pname{PolSe_{CT}}(\vec{U})\parallel\pname{CNOT}(\vec{V})\parallel\pname{MMT}(\vec{W}))
\end{array}
\]
where each process is parameterised by their respective list of the channels ($\vec{X},\vec{U},\vec{V}$ and $\vec{W}$) on which it interacts with other processes. $\vec{X}$ contains channels $a,b,out_{1},cnt$ and $out_{2}$. $\vec{U}$ contains $a,b,c,d,e,f$ and $\vec{W}$ contains $k,l,q,r,out_{1},cnt,out_{2}$.
The scope of the list of channels ($\vec{Y}$) is restricted, indicated by $\new$ in the definition. $\vec{Y}$ comprises of the channels $c,d,e,f,g,h,m,l,k,o,q,r,u$ and $v$. We have omitted the types from our definitions, for brevity. Also, the definitions include a list of channels rather than individual channel names. The CQP definition for $\ppolsect$ is
$
\pname{PolSe_{CT}}(\vec{U}) = \pname{PolSe}(a,c,d)\parallel\pname{PolSe}(b,e,f).
$
Recall from Section~\ref{sec:LOCQP} that $\ppolse$ represents the combination of a PBS and PR. 

\begin{figure}[t]
\begin{tikzpicture}[scale=0.35]
\draw (0,0) -- (2,0);
\draw[fill=white] (2,-1) rectangle (5,1);\draw (5,0.5) -- (9,0.5);\draw (9,0.5) -- (11,1.75);\draw (5,-0.5) -- (9,-0.5); \draw (9,-0.5) -- (11,-2.5); \node at (3.5,0){$\bold{_{PolSe}}$}; \node at (1,0.5){$\bold{_{a}}$}; \node at (6,1){$\bold{_{c}}$}; \node at (6,0){$\bold{_{d}}$};
\draw (0,-3) -- (2,-3);
\draw[fill=white] (2,-4) rectangle (5,-2);\draw (5,-3.5) -- (7,-3.5);\draw (5,-2.5) -- (7,-2.5); \node at (3.5,-3){$\bold{_{PolSe}}$};\node at (1,-2.5){$\bold{_{b}}$}; \node at (6,-2){$\bold{_{e}}$}; \node at (6,-3){$\bold{_{f}}$};
\draw[fill=white] (7,-4) rectangle (9,-2);\draw (9,-3.5) -- (11,-3.5);\draw (9,-2.5) -- (11,-0.5); \node at (8,-3){$\bold{_{BS_{1}}}$}; \node at (10.5,-0.5){$\bold{_{g}}$}; \node at (10,-3){$\bold{_{h}}$};
\draw[fill=white] (11,0.25) rectangle (13,2.25);\draw (13,1.75) -- (15,0.5);\draw (15,0.5) -- (19,0.5);\draw (13,0.5) -- (15,1.75); \node at (12,1.25){$\bold{_{BS_{2}}}$};\draw (9,1.75) -- (11,0.5); \node at (8.7,2){$\bold{_{i}}$}; \node at (15.5,2){$\bold{_{j}}$};\node at (13.5,1.95){$\bold{_{k}}$};
\draw[fill=white] (11,-0.25) rectangle (13,-2.75);\draw (13,-2.5) -- (15,-0.5);\draw (15,-0.5) -- (19,-0.5);\draw (13,-0.5) -- (15,-2.5); \node at (12,-1.5){$\bold{_{BS_{3}}}$};\node at (13.5,-2.5){$\bold{_{l}}$};\node at (13.75,-0.5){$\bold{_{m}}$};
\draw[fill=white] (11,-3.25) rectangle (13,-5.25);\draw (13,-3.5) -- (15,-3.5);\draw (13,-4.75) -- (15,-4.75); \node at (12,-4.25){$\bold{_{BS_{4}}}$};\draw (9,-4.75) -- (11,-4.75);\node at (10,-4.25){$\bold{_{n}}$};\node at (14.5,-3.1){$\bold{_{o}}$};\node at (14,-4.25){$\bold{_{p}}$};
\draw[fill=white] (15,-4) rectangle (17,-2);\draw (17,-3.5) -- (19,-3.5);\draw (17,-2.5) -- (19,-2.5); \node at (16,-3){$\bold{_{BS_{5}}}$};\node at (18,-2){$\bold{_{q}}$};\node at (18,-3){$\bold{_{r}}$};
\draw[fill=white] (19,-1) rectangle (21,1);\draw (21,0) -- (23,0);\draw (21,-0.25) -- (23,-0.25); \node at (20,0){$\bold{_{Det_{1}}}$};\node at (22,0.5){$\bold{_{u}}$};
\draw[fill=white] (19,-2) rectangle (21,-4);\draw (21,-3) -- (23,-3);\draw (21,-2.75) -- (23,-2.75); \node at (20,-3){$\bold{_{Det_{2}}}$};\node at (22,-2.4){$\bold{_{v}}$};
\draw[fill=white] (23,1) rectangle (27,-4);\node at (25,-1.5){$\bold{_{Counter}}$};\draw (27,-0.5) -- (29,-0.5);\draw (27,-0.25) -- (29,-0.25); \node at (28.5,0){$\bold{_{out_{1}}}$};\draw (27,-3.5) -- (29,-3.5);\draw (27,-3.25) -- (29,-3.25); \node at (28.5,-3){$\bold{_{out_{2}}}$};\draw (27,-2) -- (29,-2);\draw (27,-1.75) -- (29,-1.75); \node at (28.5,-1.5){$\bold{_{cnt}}$};
\draw [dashed, thick, black] (0.5,2.75) rectangle (6.35,-5.5);\node at (3.5,-6.5){$\bold{_{PolSe_{CT}}}$};
\draw [dashed, thick, blue] (6.5,2.75) rectangle (18.5,-5.5);\node at (12,-6.5){$\bold{_{CNOT}}$};
\draw [dashed, thick, red] (18.75,2.75) rectangle (27.5,-5.5);\node at (22,-6.5){$\bold{_{MMT}}$};
\end{tikzpicture}
\caption{Model of LOQC CNOT gate: The dashed lines enclose the subsystems which are defined in the text.}
\label{fig:mod}
\end{figure}
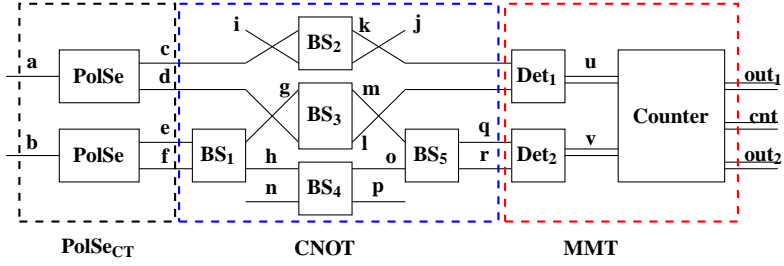

Each BS is represented by a process $\pbs$ and is annotated to show the correspondence with Figure~\ref{fig:mod}. $BS_{2}$ and $BS_{3}$ have their inputs crossed over, corresponding to their orientation \cite{Brien2003}. Vacuum states $y$ and $z$ are created by $(\ns y,z)$ and communicated to $BS_{2}$ and $BS_{4}$ respectively through the channels $i$ and $n$. These $BSs$ combine to form $\pcnot$ defined as:
\[
\begin{array}{rcl}
\pname{CNOT}(\vec{V}) =  (\new g,h,m,o,i,j,n,p)(\ns y,z)(\pname{BS_{1}}(e,f,g,h,\frac{1}{2})\parallel\outp{i}{y}\sep\nil\parallel\pname{BS_{2}}(i,c,k,j,\frac{1}{3})\parallel\inp{j}{y}\sep\nil\parallel\\\pname{BS_{3}}(g,d,m,l,\frac{1}{3})\parallel\outp{n}{z}\sep\nil\parallel\pname{BS_{4}}(h,n,o,p,\frac{1}{3})\parallel\inp{p}{z}\sep\nil\parallel\pname{BS_{5}}(m,o,q,r,\frac{1}{2}))
\end{array}
\]
Here $\vec{V}$ contains the channels $c,d,e,f,k,l,q$ and $r$. The outputs of $\pcnot$ are sent through the channels $k,l,q$ and $r$, to the process $\pmmt$. The unused $BS$ outputs $j$ and $p$ are absorbed by $\inp{j}{y}$ and $\inp{p}{z}$.
\[
\begin{array}{rcl}
\pname{MMT}(\vec{W}) =  (\new u,v)(\pname{Det_{1}}(k,l,u)\parallel\pname{Det_{2}}(q,r,v)\parallel\pname{Counter}(u,v,out_{1},cnt,out_{2},b))
\end{array}
\]
$\pmmt$ performs the measurement. Detectors $\pdetone, \pdettwo$ are annotated to match Figure~\ref{fig:mod} and measure the number states associated with the control and target qubits. The output of a detector are two classical values which represents the measurement outcome, that is the number of photons detected. The outcomes of the detector processes are given as inputs to the process  $\pcounter$.
\[
\begin{array}{rcl}
\pname{Counter}(u,v,out_{1},cnt,out_{2}, b : \Bit)
  =  \inp{u}{\typed{c_{0}}{\Int},\typed{c_{1}}{\Int}}\sep\inp{v}{\typed{t_{0}}{\Int},\typed{t_{1}}{\Int}}\sep\hspace{34mm}  
  \\\outp{out_{1}}{\textrm{$\bold{if}$ $(c_{0} + c_{1} = 1)$ $\bold{then}$ $ c_{1}$ $\bold{else}$ $0$}}\sep\outp{out_{2}}{\textrm{$\bold{if}$ $(t_{0} + t_{1} = 1)$ $\bold{then}$ $ t_{1}$ $\bold{else}$ $0$}}\sep\\\outp{cnt}{\textrm{$\bold{if}$ $(c_{0} + c_{1} = 1)$ $and$ $(t_{0} + t_{1} = 1)$ $\bold{then}$ $ b = 1$ $\bold{else}$ $b = 0$}}\sep\nil\hspace{20mm}
\end{array}
\]

$\pcounter$ represents the coincidence measurement in optical experiments. Coincidence is observed by detecting two photons, one at channel $u$ and the other at $v$. It also provides the correct output of the CNOT gate in terms of classical bits through the channels $out_{1}$ and $out_{2}$. The coincidence count ($b$) is recorded as 
 $1$ at the output of the channel $cnt$. The unsuccessful outcomes of the CNOT gate are recorded as $0$ at the three output channels. This is determined by the $\bold{if \dots else}$ conditions in the definition. When we consider the correctness of the system, we will prove that $\pmodel$ is equivalent to the following $\pspec$ process. We use the same process $\ppolsect$ as the input for $\pspec$. 
\[
\begin{array}{rcl}
\pname{Specification_1}(a,b,out_{1},cnt,out_{2}) 
=  (\new c,d,e,f,g)(\pname{PolSe_{CT}}(\vec{U})\parallel\pname{OP}(\vec{C})\parallel\pname{Output}(\vec{D}))
\end{array}
\]
There may be other ways of expressing the specification, for example without converting the polarisation qubit into the spatial encoding, but we do not investigate them in the present paper. Here, $\vec{C}$ is a list of channels containing $c,d,e,f,g,h,i,j,k$ and $\vec{D}$ consists $g,h,i,j,k,out_{1},cnt,out_{2}$. $\pop$ performs the CNOT operation with a certain probability and is defined by
\[
\begin{array}{rcl}
\pname{OP}(c,d,e,f,g,h,i,j,k)
=  (\qbit : q_{2})\sep\inp{c}{s_{0}}\sep\inp{d}{s_{1}}\sep\inp{e}{s_{2}}\sep\inp{f}{s_{3}}\sep\action{s_{2},s_{3}\trans\qgate{H}}\sep\action{q_{2}\trans\qgate{U_{\frac{1}{9}}}}\sep\hspace{9mm}\\\action{(s_{0},s_{1}),(s_{2},s_{3})\trans\qgate{CZ}}\sep\action{s_{2},s_{3}\trans\qgate{H}}\sep\outp{h}{s_{0}}\sep\outp{i}{s_{1}}\sep\outp{j}{s_{2}}\sep\outp{k}{s_{3}}\sep\outp{g}{\measure q_{2}}\sep\nil
\end{array}
\]
$\pop$ possesses a qubit $q_{2}$ (initialised to $\ket{0}$). A random bit is generated with certain probability ($\frac{8}{9}$ for bit $0$) by measuring $q_{2}$ after the unitary operation with $U_{\frac{1}{9}}$. This is followed by a series of unitary operations namely Hadamard operation ($\mathsf{H}$) which is applied twice on a pair of number states ($s_{2}$,$s_{3}$) and a controlled Z ($\mathsf{CZ}$) where $s_{0},s_{1}$ acts as the control pair and $s_{2},s_{3}$ is the target pair. The combination of a $\mathsf{H}$, $\mathsf{CZ}$ and another $\mathsf{H}$ constitutes a CNOT, which is an abstract version of the number state computation. The theory of these operators for number states are not discussed in this paper but are provided in \cite{Myers2005}. The data are then communicated to the process $\poutp$:
\[
\begin{array}{rcl}
\pname{Output}(g,h,i,j,k,out_{1},cnt,out_{2}) 
=  \inp{g}{\typed{x}{\bit}}\sep\inp{h}{s_{0}}\sep\inp{i}{s_{1}}\sep\inp{j}{s_{2}}\sep\inp{k}{s_{3}}\sep\outp{cnt}{x}\sep\hspace{36mm}\\\outp{out_{1}}{\textrm{$\bold{if}$ $(x = 1)$ $\bold{then}$ $ \measure s_{1}$ $\bold{else}$ $0$}}\sep\outp{out_{2}}{\textrm{$\bold{if}$ $(x = 1)$ $\bold{then}$ $ \measure s_{3}$ $\bold{else}$ $0$}}\sep\nil
\end{array}
\]
This gives the correct output in the form of classical bits of the CNOT operation when $x$ equals one, which is artificially making the specification work with a certain probability ($\frac{1}{9}$). When $x$ equals zero, the specification does not work and we get zero at all the output channels.

\subsection{Semantics of CQP}

In the previous section, we have described informally the behaviour of the processes which represent the linear optical elements of the CNOT gate model. In this section we will explain the formal semantics of CQP although without giving all of the definitions. Full definitions are in the Appendix. The execution of a system is not completely described by the process term (which is the case for classical process calculus) but also depends on the quantum state. Hence the operational semantics are defined using \emph{configurations}, which represent both the quantum state and the process term. 
\begin{Definition}[Configuration]
A configuration is defined as a tuple of the form $(\tilde{x}: \tilde{T};\sigma;\omega;P)$ where $\tilde{x}$ is a list of names (qubits $\tilde{q}$, number states $\tilde{s}$ or both) associated with their types $\tilde{T}$, $\sigma$ is a mapping from names ($\tilde{x}$) to the quantum state and $\omega$ is a list of names associated with the process P
\end{Definition}
We operate with configurations such as 
$(q_{1}: \Qbit, s_{0}: \NS, s_{1}: \NS; \qstore{q_{1},s_{0},s_{1}}{(\ket{0}\ket{10}+ \ket{1}\ket{01})};q_{1};\outp{c}{q_{1}}\sep P)$

This configuration means that the global quantum state consists of a qubit, $q_{1}$, number states $s_{0}$ and $s_{1}$, in the specified state; that the process term under consideration has access to qubit $q_{1}$ but not to the number states; and that the process itself is $\outp{c}{q_{1}}\sep P$. 
The semantics of CQP consists of labelled transitions between configurations which is essentially required for the equivalence of processes. We now present the complete \emph{labelled transition rules} of CQP that are extended from the previous work in order to verify LOQC, which is the focus of this paper.

$\bold{Expression}$ $\bold{Transition}$ $\bold{Rules:}$
\begin{figure}[ !t]
\scriptsize
  \begin{gather*}
    \tag{\Rplus}
    (\tilde{x}:\tilde{T};\sigma; \omega; u + v) \transitionv (\tilde{x}:\tilde{T};\sigma; \omega;w) \textrm {if $u$ and $v$ are integer literals and $w = u + v$}\\
    \tag{\Rnsmeasuret}
 (\tilde{x}:\tilde{T};\qstore{\tilde{x}}{\sum_{\tilde{s} \geq 0}\alpha_{\tilde{s}}\ket{\beta_{\vec{s}}}\ket{\vec{s}}}; \omega; \measure{s_{a},s_{b}})  \transitionv 
   \oplus_{k,l \ge 0}g_{kl}(\tilde{x}:\tilde{T}; \qstore{\tilde{x}}{\sum_{\tilde{s'}\geq 0}\frac{\alpha_{\tilde{s'}}}{\sqrt{g_{kl}}}\ket{\beta_{\vec{s'}}}\ket{\vec{s'}}}; \omega; \lambda yz\bullet (y,z); k,l)
 \textrm{$   $  where $g_{kl} =  \sum_{\tilde{i}}|\alpha_{\tilde{s'}}|^{2},$}\\ \textrm{$ \tilde{s} =  s_{0},\dots, s_{n-1}, \tilde{s'} = s_{0},\dots,s_{a-1},k,\dots,l,s_{b+1},\dots,s_{n-1}, \tilde{i} = s_{0},\dots,s_{n-1} \setminus (s_{a},s_{b})$}\textrm{and $(a ,b)\in\{0,\dots,n-1\}$ and $a \neq b$}\\
   \tag{\Rpsmeasure}
 (\tilde{x}:\tilde{T};\qstore{\tilde{x}}{\sum_{\tilde{s} \geq 0}\alpha_{\tilde{s}}\ket{\beta_{\vec{s}}}\ket{\vec{s}}}; \omega; \psmeasure{s_{a},s_{b}})  \transitionv 
   \oplus_{k,l \in \{0,1\}, k \neq l} h_{kl}(\tilde{x}:\tilde{T}; \qstore{\tilde{x}}{\sum_{\tilde{s'} \geq 0}\frac{\alpha_{\tilde{s''}}}{\sqrt{h_{kl}}}\ket{\beta _{\vec{s'}}}\ket{\vec{s'}}}; \omega; \lambda z\bullet z; l)
 \textrm{$   $  where $ h_{kl} = \sqrt{g_{op}}\frac{1}{\sum_{\tilde{j}}|\alpha_{\tilde{s''}}|^{2}}$} \\ \textrm{ and $g_{op} =  \sum_{\tilde{i}}|\alpha_{\tilde{s'}}|^{2}$, $o,p \ge 0,$ $ \tilde{s} =  s_{0},\dots, s_{n-1}, \tilde{s'} = s_{0},\dots,s_{a-1},o,\dots,p,s_{b+1},\dots,s_{n-1},$} \\ \textrm{ $ \tilde{i} = s_{0},\dots,s_{n-1} \setminus (s_{a},s_{b})$ $\tilde{s''} = s_{0},\dots,s_{a-1},k,\dots,l,s_{b+1},\dots,s_{n-1},$} \\ \textrm{and $\tilde{j} = s_{0},\dots,s_{a-1},k,\dots,l,s_{b+1},\dots,s_{n-1}$ and $(a ,b)\in\{0,\dots,n-1\}$ and $a \neq b$}\\
  \tag{\Rmeasureq}
      ( q_0,\dots,q_{n-1} = \alpha_0 \ket{\phi_0} + \dots + \alpha_{2^n-1} \ket{ \phi_{2^n-1}}; \omega; \measure{ q_0,\dots, q_{r-1} } ) \transitionv 
      \Dist{0 \le m < 2^r}{g_m} ( q_0, \dots, q_{n-1} = \frac{\alpha_{l_m}}{\sqrt{g_m}} \ket{ \phi_{l_m}} + \dots + \frac{ \alpha_{u_m}}{\sqrt{g_m}} \ket{ \phi_{u_m}}; \omega; \ltrm{x}{x}{m} )\\
    \textrm{where } l_m = 2^{n-r}m, u_m = 2^{n-r}(m+1)-1, g_m = |\alpha_{l_m}|^2 + \dots + |\alpha_{u_m}|^2 \\
   \tag{\Rnstrans}
    (\tilde{q}:\Qbit,\tilde{s}:\NS; \qstore{\tilde{q},\tilde{s}}{\ket{\psi}} ; \omega; s_{0},\dots,s_{2r-1}\trans{U}) \transitionv 
    (\tilde{q}:\Qbit, \tilde{s}:\NS; \qstore{\tilde{q},s_0,\dots,s_{n-1}}{(I_{|\vec{q}|}\otimes U \otimes I_{(n-2r)}) \ket{\psi}}; \omega; \unit)\\ 
        \tag{\Rifthen}
    (\tilde{x}:\tilde{T};\sigma; \omega; \textrm{if true then $e$ else $e'$}) \transitionv (\tilde{x}:\tilde{T};\sigma; \omega;e) \textrm{ or }  (\tilde{x}:\tilde{T};\sigma; \omega; \textrm{if false then $e'$ else $e$}) \transitionv (\tilde{x}:\tilde{T};\sigma; \omega;e')\\
      \tag{\Rcontext}
    \begin{prooftree}
      \forall i \in I. (\vec{x}:\vec{T}; \qstore{\vec{x}}{\ket{\psi_i}} ; \omega; e\{\vec{u_i}/\vec{y}\}) \transitionv \Dist{j \in J_i}{g_{ij}} (\vec{x}:\vec{T}; \qstore{\vec{x}}{\ket{\psi_{ij}}} ; \omega; \ltrm{\vec{z}}{e'\{\vec{u_i}/\vec{y}\}}{\vec{v}_{ij}} )
      \justifies
      \Dist{i \in I}{h_i} (\vec{x}:\vec{T}; \qstore{\vec{x}}{\ket{\psi_i}} ; \omega; \ltrm{\vec{y}}{E[e]}{\vec{u_i}}) \transitione \Dist{\substack{i \in I\\j \in J_i}}{h_ig_{ij}} (\vec{x}:\vec{T}; \qstore{\vec{x}}{\ket{\psi_{ij}}} ; \omega; \ltrm{\vec{y}\vec{z}}{E[e']}{\vec{u_i},\vec{v}_{ij}})
    \end{prooftree}
  \end{gather*}
  \normalsize
  \caption{Transition rules for values and expressions.}
\label{fig:trans_mixed_exp}
\end{figure}
For the evaluation of expressions we also introduce \emph{expression configurations} $(\tilde{x}: \tilde{T};\sigma;\omega;e)$, which are similar to configurations, but include an expression in place of the process. The semantics of expressions is defined by the reduction relations $\transitionv$ (on values) and $\transitione$ (on expressions), given in Figure~\ref{fig:trans_mixed_exp}. Rule $\Rplus$ deal with the evaluation of terms that result in values. Rules $\Rnsmeasuret$,$\Rpsmeasure$ and $\Rmeasureq$ are measurement rules which produces a mixed configuration. First two rules measure a pair of number states and the last rule measures qubit. $\Rnsmeasuret$ produces a mixed configuration over the possible measurement outcomes $k$ and $l$. The measurement outcomes are classical values which are the number of photons detected. $\Rpsmeasure$ is a \emph{post-selective} measurement rule which produces a mixed configuration over the possible measurement outcome $l$. Rule $\Rnstrans$ deals with unitary transformations which result in literal unit. The important aspect of $\Rnstrans$ and the measurement rules is the effect they have on the quantum state. 

The rule $\Rcontext$ has two primary purposes; it is used for the evaluation of expressions in an expression context $E$ and it is also used of the evaluation of expressions in mixed configurations. The evaluation of a mixed expression configuration $\Dist{i \in I}{h_i}(\vec{x}:\vec{T}; \qstore{\vec{x}}{\ket{\psi_i}} ; \omega; \ltrm{\vec{y}}{E[e]}{\vec{u_i}})$ is determined by the evaluation of each component. For a given component, the pure expression configuration is obtained by substitution of the respective values; $(\vec{x}:\vec{T}; \qstore{\vec{x}}{\ket{\psi_i}} ; \omega; E[e]\{\vec{u_i}/\vec{y}\})$. For this configuration we isolate the context and consider the evaluation of $e\{\vec{u_i}/\vec{y}\})$. The resulting configuration may be a mixed expression configuration with new variables $\vec{z}$ introduced; specifically we end up with a term $\ltrm{\vec{z}}{e'\{\vec{u_i}/\vec{y}\}}{\vec{v}_{ij}}$ where, due to the use of the substitution, $e'$ is constant across each $i$. The results for each $i$ are combined to give the final term $\ltrm{\vec{y}\vec{z}}{E[e']}{\vec{u_i},\vec{v}_{ij}}$ incorporating variables $\vec{z}$ and $\vec{y}$.
\begin{figure*}[t]
\scriptsize
 \begin{gather*}
  \tag{\Pout}
    (\vec{p},\vec{q}: \tilde{\Qbit}, \vec{r},\vec{s}: \tilde{\NS},\qstore{\vec{p}\vec{q}\vec{r}\vec{s}}{\ket{\psi}} ; \vec{p},\vec{q},\vec{r},\vec{s}; \outp{c}{\vec{v},\vec{q},\vec{s}}.P) \transitionp{\outp{c}{\vec{v},\vec{q},\vec{s}}} (\vec{p},\vec{q}: \tilde{\Qbit}, \vec{r},\vec{s}: \tilde{\NS},\qstore{\vec{p}\vec{q}\vec{r}\vec{s}}{\ket{\psi}} ; \vec{p},\vec{r}; P)\\
    \tag{\Pin}
    (\vec{q}: \tilde{\Qbit},\vec{s}: \tilde{\NS},\qstore{\vec{q}\vec{s}}{\ket{\psi}} ; \omega; \inp{c}{\vec{y},\vec{x}}.P ) \transitionp{\inp{c}{\vec{v},\vec{p},\vec{r}}} (\vec{q}: \tilde{\Qbit},\vec{s}: \tilde{\NS}, \qstore{\vec{q}\vec{s}}{\ket{\psi}} ; \omega,\vec{p},\vec{r}; P\{\vec{v},\vec{r}/ \vec{y},\vec{p}/ \vec{x}\} ) \\
    \tag{\Ppar}
    \begin{prooftree}
      (\vec{x}: \tilde{T}, \qstore{\vec{x}}{\ket{\psi}} ; \omega; P) \transitionp{\alpha} (\vec{x}: \tilde{T},\qstore{\vec{x}}{\ket{\psi}} ; \omega'; P')
      \justifies
      (\vec{x}: \tilde{T},\qstore{\vec{x}}{\ket{\psi}} ; \omega; P \parcomp Q) \transitionp{\alpha} (\vec{x}: \tilde{T},\qstore{\vec{x}}{\ket{\psi}} ; \omega'; P' \parcomp Q)
    \end{prooftree} \\
    \tag{\Psum}
    \begin{prooftree}
      (\vec{x}: \tilde{T}, \qstore{\vec{x}}{\ket{\psi}} ; \omega; P) \transitionp{\alpha} (\vec{x}: \tilde{T}, \qstore{\vec{x}}{\ket{\psi}} ; \omega'; P')
      \justifies
      (\vec{x}: \tilde{T}, \qstore{\vec{x}}{\ket{\psi}} ; \omega; P + Q) \transitionp{\alpha} (\vec{x}: \tilde{T}, \qstore{\vec{x}}{\ket{\psi}} ; \omega'; P')
    \end{prooftree} \\
    \tag{\Pnew}
    \qquad \begin{prooftree}
      (\vec{x}: \tilde{T}, \qstore{\vec{x}}{\ket{\psi}} ; \omega; P) \transitionp{\alpha} (\vec{x}: \tilde{T}, \qstore{\vec{x}}{\ket{\psi}} ; \omega; P')
      \justifies
      (\vec{x}: \tilde{T}, \qstore{\vec{x}}{\ket{\psi}} ; \omega; (\new c)P) \transitionp{\alpha} (\vec{x}: \tilde{T}, \qstore{\vec{x}}{\ket{\psi}} ; \omega; (\new c)P')
    \end{prooftree} \quad
     \textrm{ if $\alpha \notin \{\inp{c}{\cdot}, \outp{c}{\cdot}\}$}
   \\ \tag{\Pps}
       (\tilde{x},\tilde{y}: \tilde{\Qbit}, q:\Qbit, \tilde{z}: \tilde{\NS}; \qstore{\tilde{x},q,\tilde{y},\tilde{z}}{\ket{\phi}};\omega; \{s,t\trans\qgate{PS}(q)\}\sep P) \transition{\tau}(\tilde{x},\tilde{y} : \tilde{\Qbit},\tilde{z}: \tilde{\NS}, s :\NS, t :\NS;\qstore{\tilde{x},\tilde{y},\tilde{z},s,t}{\ket{\psi}};\omega';P)
   \\ \textrm{ where $\ket{\phi} = \ket{\alpha}\ket{0}\ket{\beta}\ket{\gamma} + \ket{\alpha'}\ket{1}\ket{\beta'}\ket{\gamma'}$ , $\ket{\psi} = \ket{\alpha}\ket{\beta}\ket{\gamma}\ket{10} + \ket{\alpha'}\ket{\beta'}\ket{\gamma'}\ket{01}$} \textrm{$q \in \omega$ and $s, t \notin \omega$, $q \notin \omega'$ and $s, t \in \omega'$ }
  \end{gather*}
  \normalsize
  \caption{Transition rules for pure process configurations.}
\label{fig:trans_pure}
\end{figure*}

$\bold{Pure}$ $\bold{Configuration}$ $\bold{Transition}$ $\bold{Rules:}$
The rules for pure process configurations are given in Figure~\ref{fig:trans_pure}. This defines the input and output transitions for pure configurations. It is used in the hypothesis of $\Loutq$, $\Loutns$ and $\Lcom$ to determine the actions of the individual components in a mixed configurations. The rules namely choice ($\Psum$), parallel ($\Ppar$) and restriction ($\Pnew$) are required to define input and output actions for arbitrary process constructions. These rules are applicable for both qubits and number states and $\Pps$ is for the conversion of polarisation qubit to the number states.

\begin{figure*}
\scriptsize
  \begin{gather*} 
    \tag{\Lprob}  
    \Prob{j}{p_j} (\Dist{i}{g_i} (\vec{x}:\tilde{T};\sigma_i; \omega; P_i ) ) \ptrans{p_i} \Dist{i}{g_i} ( \vec{x}:\tilde{T};\sigma_i; \omega; P_I ) \\[2mm]
   \tag{\Lin}  
    \Dist{i}{g_i} (\vec{x}:\vec{T};\sigma_i; \omega; \ltrm{\vec{z}}{\inp{c}{\vec{q},\vec{s}}.P}{\vec{v_i} ) \transition{\inp{c}{\vec{p},\vec{r}}} \Dist{i}{g_i} (\vec{x}:\vec{T}; \sigma_i; \omega,\vec{r},\vec{p}; \ltrm{\vec{z}}{P\{\vec{p}/\vec{q},\vec{r}/\vec{s}\}}{\vec{v_i}}} ) 
  \\    \forall i \in I. ((\vec{p},\vec{q}):\vec{\Qbit},\vec{s}:\vec{\NS};\qstore{\vec{p}\vec{q}\vec{s}}{\ket{\alpha_i}\ket{\beta}}; \vec{p},\vec{s}; P\subst{\vec{v}_i}{\vec{x}} ) \transitionp{\outp{c}{\vec{u}_i,\vec{r}}}((\vec{p},\vec{q}):\vec{\Qbit},\vec{s}:\vec{\NS};\qstore{\vec{p}\vec{q}\vec{s}}{\ket{\alpha_i}\ket{\beta}}; \vec{p}',\vec{s}; P'\subst{\vec{v}_i}{\vec{x}} )\\
      \rule{356pt}{1pt} \tag{\Loutq}\\
      \Dist{i \in I}{g_i} ((\vec{p},\vec{q}):\vec{\Qbit},\vec{s}:\vec{\NS};\qstore{\vec{p}\vec{q}\vec{s}}{\ket{\alpha_i}\ket{\beta}}; \vec{p},\vec{s}; \ltrm{\vec{x}}{P}{\vec{v}_i} ) \transition{\outp{c}{U,\vec{r}}} \Prob{j \in J}{p_j} (\Dist{i \in I_j}{\frac{g_i}{p_j}} ((\vec{p'},\vec{q}):\vec{\Qbit},\vec{s}:\vec{\NS};\qstore{\vec{p}'\vec{r}\vec{q}\vec{s}}{\Pi\ket{\alpha_i}\ket{\beta}}; \vec{p}',\vec{s}; \ltrm{\vec{x}}{P'}{\vec{v}_i} ) ) \\
    \textrm{where } U = \{\vec{u}_i ~|~ i \in I\} = \{\vec{w}_{j} ~|~ j \in J\} \textrm{ and } \forall j \in J, I_j = \{i | \vec{u}_i = \vec{w}_{j}\}, p_j = \sum_{i \in I_j} g_i \\
     \textrm{ and } \vec{r} \subseteq \vec{p}, \vec{p}' = \vec{p} \setminus \vec{r}, \Pi \textrm{ corresponds to the permutation } \pi:\vec{p}\vec{q}\vec{s} \mapsto \vec{p}'\vec{r}\vec{q}\vec{s} ~. 
   \\   \forall i,j \in I. (\vec{p}:\vec{\Qbit},(\vec{t},\vec{s}):\vec{\NS};\qstore{\vec{p}\vec{t}\vec{s}}{\ket{\alpha}\ket{\beta_{ij}}}; \vec{p},\vec{s}; P\{\vec{v}_{ij}/ \vec{x},\vec{w}_{ij}/ \vec{y}\} ) \transitionp{\outp{c}{\vec{u}_{ij},\vec{r}}} (\vec{p}:\vec{\Qbit},(\vec{t},\vec{s}):\vec{\NS};\qstore{\vec{p}\vec{t}\vec{s}}{\ket{\alpha}\ket{\beta_{ij}}}; \vec{p},\vec{s}'; P'\{\vec{v}_{ij}/ \vec{x},\vec{w}_{ij}/ \vec{y}\})\\
     \rule{356pt}{1pt}  \tag{\Loutns}\\ 
      \Dist{i,j \in I}{g_{ij}} (\vec{p}:\vec{\Qbit},(\vec{t},\vec{s}):\vec{\NS};\qstore{\vec{p}\vec{t}\vec{s}}{\ket{\alpha}\ket{\beta_{ij}}}; \vec{p},\vec{s}; \ltrm{\vec{x}\vec{y}}{P}{\vec{v}_{ij},\vec{w_{ij}}} ) \transition{\outp{c}{U,\vec{r}}}\\ \Prob{k \in J}{p_{k}} (\Dist{i,j \in I_k}{\frac{g_{ij}}{p_{k}}} (\vec{p}:\vec{\Qbit},(\vec{t},\vec{s}':\vec{\NS});\qstore{\vec{p}\vec{t}\vec{s}'\vec{r}}{\Pi\ket{\alpha}\ket{\beta_{ij}}}; \vec{p},\vec{s}'; \ltrm{\vec{x}\vec{y}}{P'}{\vec{v}_{ij},\vec{w}_{ij}} ) ) \\
   \textrm{where } U = \{\vec{u}_{ij} ~|~ i,j \in I\} = \{\vec{e}_{k} ~|~ k \in J\}, 
   \textrm{ and } \forall k \in J, I_{k} = \{i,j | \vec{u}_{ij} = \vec{e}_{k}\}, p_{k} = \sum_{i,j \in I_{k}} g_{ij} \\
    \textrm{ and } \vec{r} \subseteq \vec{s}, \vec{s}' = \vec{s} \setminus \vec{r}, \Pi \textrm{ corresponds to the permutation } \pi:\vec{p}\vec{t}\vec{s} \mapsto \vec{p}\vec{t}\vec{r}\vec{s}' ~. 
    \\[0mm]
    \tag{\Lcom} 
    \begin{prooftree}
      \begin{array}{l}
	\forall i \in I. (\vec{x}:\vec{T};\sigma_i; \omega,\vec{r}; P \subst{\vec{v_i}}{\vec{z}} ) \transitionp{\outp{c}{\vec{u}_i,\vec{r}}} (\vec{x}:\vec{T};\sigma_i; \omega; P'\subst{\vec{v_i}}{\vec{z}} ) \\
	\forall i \in I. (\vec{x}:\vec{T};\sigma_i; \omega; Q \subst{\vec{v_i}}{\vec{z}} ) \transitionp{\inp{c}{\vec{u}_i,\vec{r}}} (\vec{x}:\vec{T};\sigma_i; \omega,\vec{r}; Q'\subst{\vec{v_i}}{\vec{z}} )
      \end{array}
      \justifies
      \Dist{i \in I}{g_i} (\vec{x}:\vec{T};\sigma_i; \omega,\vec{r}; \ltrm{\vec{z}}{P \parcomp Q}{\vec{v_i}} ) \transition{\tau} \Dist{i \in I}{g_i} (\vec{x}:\vec{T};\sigma_i; \omega,\vec{r}; \ltrm{\vec{z}}{P' \parcomp Q'}{\vec{v_i}} )
    \end{prooftree} \\[2mm]
    \tag{\Lpar}
    \begin{prooftree}
      \Dist{i \in I}{g_i} (\vec{x}:\vec{T};\sigma_i; \omega; \ltrm{\vec{z}}{P}{\vec{v_i}}) \transition{\alpha} \Dist{\substack{i \in I\\j \in J_i}}{g_ih_{ij}} (\vec{x}:\vec{T};\sigma_{ij}; \omega'; \ltrm{\vec{z}\vec{y}}{P'}{\vec{v_i},\vec{w_{ij}}})
      \justifies
      \Dist{i \in I}{g_i} (\vec{x}:\vec{T};\sigma_i; \omega; \ltrm{\vec{z}}{P \parcomp Q}{\vec{v_i}}) \transition{\alpha} \Dist{\substack{i \in I\\j \in J_i}}{g_ih_{ij}} (\vec{x}:\vec{T};\sigma_{ij}; \omega'; \ltrm{\vec{z}\vec{y}}{P' \parcomp Q}{\vec{v_i},\vec{w_{ij}}})
    \end{prooftree} \\[2mm]
    \tag{\Lsum}
    \begin{prooftree}
      \Dist{i \in I}{g_i} (\vec{x}:\vec{T};\sigma_i; \omega; \ltrm{\vec{z}}{P}{\vec{v_i}}) \transition{\alpha} \Dist{\substack{i \in I\\j \in J_i}}{g_ih_{ij}} (\vec{x}:\vec{T};\sigma_{ij}; \omega'; \ltrm{\vec{z}\vec{y}}{P'}{\vec{v_i},\vec{w_{ij}}})
      \justifies
      \Dist{i \in I}{g_i} (\vec{x}:\vec{T};\sigma_i; \omega; \ltrm{\vec{z}}{P + Q}{\vec{v_i}}) \transition{\alpha} \Dist{\substack{i \in I\\j \in J_i}}{g_ih_{ij}} (\vec{x}:\vec{T};\sigma_{ij}; \omega'; \ltrm{\vec{z}\vec{y}}{P'}{\vec{v_i},\vec{w_{ij}}})
    \end{prooftree}\\
      \tag{\Lnew}
    \begin{prooftree}
      \Dist{i \in I}{g_i} (\vec{x}:\vec{T};\sigma_i; \omega; \ltrm{\vec{z}}{P}{\vec{v_i}}) \transition{\alpha} \Dist{\substack{i \in I\\j \in J_i}}{g_ih_{ij}} (\vec{x}:\vec{T};\sigma_{ij}; \omega'; \ltrm{\vec{z}\vec{y}}{P'}{\vec{v_i},\vec{w_{ij}}})
      \justifies
      \Dist{i \in I}{g_i} (\vec{x}:\vec{T};\sigma_i; \omega; \ltrm{\vec{z}}{(\new c)P}{\vec{v_i}}) \transition{\alpha} \Dist{\substack{i \in I\\j \in J_i}}{g_ih_{ij}} (\vec{x}:\vec{T};\sigma_{ij}; \omega'; \ltrm{\vec{z}\vec{y}}{(\new c)P'}{\vec{v_i},\vec{w_{ij}}})
    \end{prooftree} \\ 
    \hspace{60mm} \textrm{if $\alpha \notin \{\inp{c}{\cdot}, \outp{c}{\cdot}\}$}  \\[2mm]
   \Dist{i \in I}{g_i} (\vec{q}:\vec{\Qbit},\vec{s}:\vec{\NS}; \qstore{\vec{q}\vec{s}}{\ket{\beta_i}\ket{\gamma_i}} ; \omega; \ltrm{\vec{z}}{(\qbit : y)P}{\vec{v_i}} ) \transition{\tau} \Dist{i \in I}{g_i} (\vec{q}:\vec{\Qbit},q:\Qbit,\vec{s}:\vec{\NS}; \qstore{\vec{q},q,\vec{s}}{\ket{\beta_i}\ket{\phi_j}\ket{\gamma_i}} ; \omega,q; \ltrm{\vec{z}}{P\subst{q}{y}}{\vec{v_i}} ) \\
    \tag{\Lqbit} \hspace{70mm} \textrm{where $q$ is fresh} \\[2mm]
       \Dist{i \in I}{g_i} (\vec{q}:\vec{\Qbit},\vec{s}:\vec{\NS}; \qstore{\vec{q}\vec{s}}{\ket{\beta_i}\ket{\gamma_i}} ; \omega; \ltrm{\vec{z}}{(\ns : y)P}{\vec{X}} ) \transition{\tau}\Dist{i \in I}{g_i} (\vec{q}:\vec{\Qbit},r:\NS,\vec{s}:\vec{\NS}; \qstore{\vec{q},r,\vec{s}}{\ket{\beta_i}\ket{\psi_j}\ket{\gamma_i}} ; \omega,r; \ltrm{\vec{m}}{P\subst{r}{y}}{\vec{v_i}} ) \\
    \tag{\Lns} \hspace{70mm} \textrm{where $r$ is fresh} \\[2mm]
    \tag{\Lact} \Dist{i \in I}{g_i} (\vec{x}:\vec{T};\sigma_i; \omega; \ltrm{\vec{z}}{\action{u}.P_i}{\vec{v_i}} ) \transition{\tau} \Dist{i \in I}{g_i} (\vec{x}:\vec{T};\sigma_i; \omega; \ltrm{\vec{z}}{P}{\vec{v_i}} ) \\[2mm]
    \tag{\Lps}
      \Dist{i \in I}{g_i} (\tilde{p},\tilde{q}: \tilde{\Qbit}, q_{c}:\Qbit, \tilde{r}: \tilde{\NS}; \qstore{\tilde{p},q_{c},\tilde{q},\tilde{r}}{\ket{\phi}};\omega; \ltrm{\vec{z}}{\action{s_{a},s_{b}\trans\qgate{PS}(q_{c})}}\sep P, \vec{v_i})\\ \transition{\tau} \Dist{i \in I}{g_i}(\tilde{p},\tilde{q} : \tilde{\Qbit},\tilde{r}: \tilde{\NS}, s_{a} :\NS, s_{b} :\NS;\qstore{\tilde{p},\tilde{q},\tilde{r},s_{a},s_{b}}{\ket{\psi}};\omega'; \ltrm{\vec{z}}P, \vec{v_i})\\
    \tag{\Lexpr} 
    \begin{prooftree}
      \Dist{i \in I}{h_i} (\vec{x}:\vec{T};\sigma_i; \omega; \ltrm{\vec{y}}{e}{\vec{v_i}} ) \transitione \Dist{\substack{i \in I\\j \in J_i}}{h_i g_{ij}} (\vec{x}:\vec{T};\sigma_{ij}; \omega; \ltrm{\vec{y}\vec{z}}{e'}{\vec{v_i},\vec{w_{ij}}} )
      \justifies
      \Dist{i \in I}{h_i} (\vec{x}:\vec{T};\sigma_i; \omega; \ltrm{\vec{y}}{F[e]}{\vec{v_i}} ) \transition{\tau} \Dist{\substack{i \in I\\j \in J_i}}{h_i g_{ij}} (\vec{x}:\vec{T};\sigma_{ij}; \omega; \ltrm{\vec{y}\vec{z}}{F[e']}{\vec{v_i},\vec{w_{ij}}} )
    \end{prooftree}
  \end{gather*}
    \normalsize
\caption{Transition rules for mixed process configurations.}
\label{fig:trans_mixed}
\end{figure*}

$\bold{Mixed}$ $\bold{Configuration}$ $\bold{Transition}$ $\bold{Rules:}$
The transition relation on mixed configurations is defined by the rules in Figure~\ref{fig:trans_mixed}. The rule $\Lprob$ is a probabilistic transition in which $p_{i}$ is the probability of the transition. The rules $\Lin$, $\Loutq$ and $\Loutns$ represent the input and output actions respectively, which are the visible interactions with the environment. $P\{\vec{v_{i}}/\vec{z}\}$ indicates that $P$ with a list of values $v_{i}$ is substituted for the list of variables $\vec{z}$. When the two processes of input and output actions are put in parallel then each has a partner for its potential interaction, and the input and output can synchronise, resulting in a $\tau$ transition which is given by the rule $\Lcom$. The rule $\Lact$ just removes actions. This is a reduction of the action expression to $v$ which would involve effects like measurement or transformation of the quantum state. Rules $\Lqbit$ and $\Lns$ are for introducing additional $\Qbit$ and $\NS$ variables respectively. $\mathsf{ns}$ declarations represents vacuum states. Since the values associated with the an input action are determined by the environment, this action is identical across all components in a mixed configuration. $\Lpar$, $\Lsum$ and $\Lnew$ can then be used to define inputs on other process constructions in a mixed configuration.

The rule $\Loutq$ and $\Loutns$ is the point at which mixed configurations are combined with probabilistic branching. Branching must occur when and only when there is information to distinguish the components. This information is represented by the classical values that are outputs, which may vary between the components. Some values may be the same, thereby requiring the relevant components to remain in a mixed configuration after the output. The purpose of $\Loutq$ and $\Loutns$ is to distribute the components according to the different values, and to assign an action label that represents the combined action of \emph{all} components. For example in transition $\Loutq$, each component has a pure transition $\transitionp{\outp{c}{\vec{u}_i,\vec{r}}}$ representing the channel and qubit names that are common to all components, and the values $\vec{u}_i$ that are specific to that component. The combined action label $\transition{\outp{c}{U,\vec{r}}}$ consists of these common elements and the set $U$ of all the value tuples.

\begin{example}
 \label{ex:cfg_measurement}
   $
  (q,s,t :\vec{T}; \qstore{q,s,t}{\alpha_{10}\ket{0}\ket{10} + \alpha_{01}\ket{1}\ket{01} +\alpha_{20}\ket{0}\ket{20}}; q,s,t; \outp{c}{\measure s,t}\sep P)$\\
$\transition{\tau}  \Dist{i,j \in \{0,1,2\}} |\alpha_{ij}|^{2}(q,s,t :\vec{T};\qstore{q,s,t}{\ket{\beta}\ket{ij}}; q,s,t; \ltrm{yz}{\outp{c}{y,z}\sep P}i,j).
 $
\end{example}
This transition represents the effect of a measurement of a pair of number states ($s$, $t$), within a process which is going to output the result of the measurement. The configuration on the left is a \emph{pure configuration}, as described before.  On the right we have a \emph{mixed configuration} in which the $\oplus$ ranges over the possible outcomes of the measurement and the $\ms{\alpha_{ij}}$ are the weights of the components in the mixture. The quantum state $\qstore{q,s,t}{\ket{\beta}\ket{ij}}$ corresponds to the measurement outcome. The expression $\ltrmshort{yz}{\outp{c}{y,z}.P}$ represents the fact that the components of the mixed configuration have the same process structure and differ only in the values corresponding to measurement outcomes. The final terms in the configuration, $i$ and $j$, shows how the abstracted variables $y$ and $z$ should be instantiated in each component. Thus the $\lambda yz$ represents a term into which expressions may be substituted, which is the reason for the $\lambda$ notation. 
The next transition ($\Rpsmeasure$) represents \emph{post-selective} measurement which filters out the measurement values that satisfies a predefined criteria. 

\begin{example}
 \label{ex:ps_cfg_measurement}
  $
  (q,s,t :\vec{T}; \qstore{q,s,t}{\alpha_{10}\ket{0}\ket{10} + \alpha_{01}\ket{1}\ket{01} +\alpha_{20}\ket{0}\ket{20}}; q,s,t; \outp{c}{\psmeasure s,t}\sep P)$
  \\
$ \transition{\tau}  \Dist{i,j \in \{0,1\}, i \ne j} |\beta_{ij}|^{2}(q,s,t :\vec{T};\qstore{q,s,t}{\ket{\delta}\ket{ij}}; q,s,t ; \ltrm{y}{\outp{c}{y}\sep P}j).
 $
\end{example}

Unlike  Example~\ref{ex:cfg_measurement}, here $i$ and $j$ can have values either $0$ or $1$ and $i \ne j$. This is the criterion for post-selection and the weights of the components in the mixture are now $\ms{\beta_{ij}}$ (where $\ms{\beta_{ij}} = \frac{\ms{\alpha_{ij}}}{\sum_{ij \in \{0,1\}}\ms{\alpha_{ij}}}$). Also, here we measure two number states $s$ and $t$, which results in one classical value.
Example~\ref{ex:ps_cfg_output} shows the effect of the output from the final configuration of Example~\ref{ex:ps_cfg_measurement}. 

\begin{example}
 \label{ex:ps_cfg_output}
 $
 \Dist{i,j \in \{0,1\}, i \ne j} |\beta_{ij}|^{2}(\vec{x}:\vec{T};\qstore{\vec{x}}{\ket{\delta}\ket{ij}}; \vec{x}; \ltrm{y}{\outp{c}{y}\sep P}i)\transition{\outp{c}{j}} \Prob{ij \in \{0,1\}, i \ne j} |\beta_{ij}|^{2}$\\
$(\vec{x}:\vec{T};\qstore{\vec{x}}{\ket{\delta}\ket{ij}}; \vec{x}; \ltrm{y}{P}j)\ptrans{\ms{\beta_{01}}}(\vec{x}:\vec{T};\qstore{\vec{x}}{\ket{1}\ket{01}}; \vec{x}; \ltrm{y}{P}1)
 $
\end{example}

Here $\vec{x}$ is a list of names consisting $q$, $s$ and $t$. The output transition produces the intermediate configuration, which is a probability distribution over pure configurations (in contrast to a mixed configuration; note the change from $\oplus$ to $\boxplus$). Because it comes from a mixed configuration, the output transition contains a \emph{set} of possible values. From this intermediate configuration there are two possible probabilistic transitions, of which one is shown ($\ptrans{\ms{\beta_{01}}}$). 

\begin{example}
 \label{ex:cfg_communication}
$
  \Dist{i,j \ge 0} g_{ij}(\vec{x}:\vec{T};\qstore{\vec{x}}{\ket{\beta}\ket{ij}}; \vec{x}; \ltrm{yz}{(\outp{c}{y}\sep P \parallel \inp{c}{y}\sep Q)}i,j)\transition{\tau}$\\ $ \Dist{i,j \ge 0} g_{ij}(\vec{x}:\vec{T};\qstore{\vec{x}}{\ket{\beta}\ket{ij}}; \vec{x}; \ltrm{yz}{( P \parallel Q)}i,j)
 $
\end{example}

Measurement outcomes may be communicated between processes without creating a probability distribution. In these cases an observer must still consider the system to be in a mixed configuration as the outcomes are communicated internally and not to the environment.

\begin{example}
 \label{ex:ps}
 $
  (q: \Qbit, r :\Qbit, p :\NS, t :\NS; \qstore{q,r,p,t}{\alpha\ket{00}\ket{10} + \beta\ket{11}\ket{01}}; q,r,p,t;
   \{u :\NS,v :\NS \trans\qgate{PS}(q)\}\sep P)\transition{\tau} (r :\Qbit,\tilde{s}':\tilde{\NS};\qstore{r,\tilde{s}'}{\alpha\ket{0}\ket{1010} + \beta\ket{1}\ket{0101}}; r,\tilde{s}'; P)
 $
\end{example}

Example~\ref{ex:ps} represents the transition $\Pps$, which is the conversion of a polarisation qubit ($q$) to the number states ($u$ and $v$). $\tilde{s}'$ indicates that it is a list of names comprising $p,t,u$ and $v$ of type $\NS$. 

\subsection{Execution of $Model_1$}

Let $t = (\emptyset;\emptyset;\emptyset;\pmodel)$ be the initial configuration. The semantics of CQP is non-deterministic and hence the transitions can proceed in different order. In the first few steps, the process $\ppolsect$ receives qubits $q_{1}$ and $q_{2}$ from the environment, constructing a global quantum state $\ket{\phi}_{q} = \alpha\ket{00} + \beta\ket{01} + \gamma\ket{10} + \delta\ket{11}$. We get the configuration,
$(q_{1} : \Qbit, q_{2} : \Qbit, q_{1}q_{2} = \ket{\phi}_{q};q_{1},q_{2};(\ppolsect ' \parallel\pcnot\parallel\pmmt))$.
After some $\tau$ transitions corresponding to $\ppolsect$ operations, the qubits are converted to the respective number states $s_{0},s_{1},s_{2}$ and $s_{3}$ by $\mathsf{PS}$ operator giving the quantum state $\ket{\phi}_{s} = \alpha\ket{1010} + \beta\ket{1001} + \gamma\ket{0110} + \delta\ket{0101}$. The configuration is now 
$(\vec{s} : \vec{\NS};\vec{s} = \ket{\phi}_{s};s_{0},s_{1},s_{2},s_{3};(\ppolsect '' \parallel\pcnot\parallel\pmmt))$.
After another set of $\tau$ transitions corresponding to the $\pcnot$ process, we get the state $\ket{\phi}_{out}$ which is given by Eq.~\ref{Eq3}. The configuration now becomes
$(\vec{s} : \vec{\NS};\vec{s} = \ket{\phi}_{out};s_{0},s_{1},s_{2},s_{3};(\pcnot' \parallel\pmmt))$.
After the measurement by both detectors, the outcomes are communicated to the $\pcounter$. This happens internally and hence, we get the mixed configuration:
 \[
 \begin{array}{l}
\Dist{\substack{ij \geq 0\\kl \geq 0}}{g_{ij}h_{ijkl}}(\vec{s} : \vec{\NS};\vec{s} = \ket{\phi_{ijkl}};s_{0},s_{1},s_{2},s_{3};\ltrm{\vec{y}}{\pcounter'}i,j,k,l)
\end{array}
\]
Here $\vec{y}$ is a list of measurement outcomes ($c_{0},c_{1},t_{0}$ and $t_{1}$). The output transitions produces the configuration below, which is a mixed state given by 
$\Dist{i,j,k,l,m \in \{0,1\}}{g_{ijm}h_{ijklm}}(\vec{s} : \vec{\NS};\vec{s} = \ket{\phi_{ijkl}};\vec{s};\ltrm{\vec{z}}{\nil}i,j,k,l,m)$
where $\vec{z}$ is $c_{1},t_{1},b$. The mixture contains both the successful and unsuccessful outcomes of $\pmodel$.

\section{Behavioural Equivalence of CQP Processes}
\label{sec-equivalence}
\label{sec:equivalence}
We now extend the theory of equivalence in CQP to apply it for LOQC. The process calculus approach to verification is to define a process $\pname{Model}$ which models the system of interest, another process $\pname{Specification}$ which expresses the specification that $\pname{Model}$
should satisfy, and then prove that $\pname{Model}$ and $\pname{Specification}$ are equivalent. 
We begin with the definition of \emph{probabilistic branching bisimilarity}, which is a congruence for CQP. 

\subsection{Probabilistic Branching Bisimilarity}

There are several types of probabilistic bisimilarity
for classical probabilistic process calculi, including
\emph{probabilistic branching bisimilarity} \cite{Trcka2008}. The equivalence
for CQP defined by Davidson \cite{DavidsonThesis}, which turns out to be a
congruence, is a form of probabilistic branching bisimilarity, adapted
to the situation in which probabilistic behaviour comes from quantum
measurement. A key point is that when considering matching of input or
output transitions involving qubits, it is the reduced density
matrices of the transmitted qubits that are required to be equal.
We will now define probabilistic branching bisimilarity in full. The 
definitions in the remainder of this section are an extension from Davidson's
thesis \cite{DavidsonThesis}.

\textbf{Notation:} Let $\opttrans{\tau}$ denote zero
or one $\tau$ transitions; let $\weaktrans{ }$ denote zero or more $\tau$
transitions; and let $\weaktrans{\alpha}$ be equivalent to $\weaktrans{ }
\transition{\alpha} \weaktrans{ }$. We write $\vec{q}$ for a list of
qubit names, and similarly for other lists.

\begin{Definition}[Density Matrix of Configurations]
  \label{def:density_matrix_configs}
  Let $\sigma_{ij} = \qstore{\vec{x}}{\ket{\psi_{ij}}}$ and $\vec{y}
  \subseteq \vec{x}$ and $t_{ij} = (\vec{x}:\vec{T};\sigma_{ij}; \omega;
  \ltrm{\vec{w}}{P}{\vec{v_{ij}}})$ and $t = \Dist{ij}{g_{ij}} t_{ij}$. Then
\[
\begin{array}{llcll}
1. & \rho(\sigma_{ij}) = \ketbra{\psi_{ij}}{\psi_{ij}} & ~~~~~~ & 4. &
\rho^{\vec{y}}(t_{ij}) = \rho^{\vec{y}}(\sigma_{ij})  \\
2. & \rho^{\vec{y}}(\sigma_{ij}) =
\trace_{\vec{x}\setminus\vec{y}}(\ketbra{\psi_{ij}}{\psi_{ij}}) & & 5. &
\rho(t) = \sum_{ij} g_{ij} \rho(t_{ij}) \\
3. & \rho(t_{ij}) = \rho(\sigma_{ij}) & & 6. & \rho^{\vec{y}}(t) = \sum_{ij} g_{ij} \rho^{\vec{y}}(t_{ij})
\end{array}
\]
\end{Definition}

We also introduce the notation $\rhoe$ to denote the reduced density
matrix of the \emph{environment} qubits or number states. Formally, if $t =
(\vec{x}:\vec{T};\qstore{\vec{x}}{\ket{\psi}}; \vec{y}; P)$ then $\rhoe(t) =
\rho^{\vec{r}}(t)$ where $\vec{r} = \vec{x} \setminus \vec{y}$. The
definition of $\rhoe$ is extended to mixed configurations in the same
manner as $\rho$. The probabilistic function $\mu: \states \times \states \rightarrow
[0,1]$ is defined in the style of \cite{Trcka2008}. It allows
non-deterministic transitions to be treated as transitions with
probability $1$, which is necessary when calculating the total
probability of reaching a terminal state.
$\mu(t,u) = \delta$ if $t \ptrans{\delta} u$; $\mu(t,u) = 1$ if $t = u$ and
$t \in \ntates$; $\mu(t,u) = 0$ otherwise.

\begin{Definition}[Probabilistic Branching Bisimulation]
  \label{def:pbb}
  An equivalence relation $\mathcal{R}$
  on configurations is a \emph{probabilistic branching bisimulation} on configurations if whenever $(t,u) \in \mathcal{R}$ the following conditions are satisfied.
  \begin{enumerate}[I.]
	    \item If $t \in \ntates$ and $t \transition{\tau} t'$ 
	      then $\exists u', u''$ such that $u \weaktrans{ } u'
              \opttrans{\tau} u''$ with $(t, u') \in \mathcal{R}$ and $(t', u'') \in \mathcal{R}$.
	    \item If $t \transition{\outp{c}{V,\vec{X}_1}} t'$ where $t' = \Prob{j \in \{1\dots m\}}{p_j} t_j'$ and $V = \{\vec{v}_1,\dots,\vec{v}_m\}$ and $\vec{X}_1$ is either $\vec{q}_1$ or $\vec{s}_1$ then $\exists u', u''$ such that $u \weaktrans{ } u' \transition{\outp{c}{V,\vec{X}_2}} u''$ with
	      \begin{enumerate}[a)]
		\item $(t, u') \in \mathcal{R}$,
		\item $u'' = \Prob{j \in \{1\dots m\}}{p_j} u_j''$,
                \item for each $j \in \{1,\dots,m\}$, $\rhoe(t_j') = \rhoe(u_j'')$.
		\item for each $j \in \{1,\dots,m\}$, $(t_j', u_j'') \in \mathcal{R}$.
	      \end{enumerate}
        \item If $t \transition{\inp{c}{\vec{v}}} t'$ then $\exists
          u', u''$ such that $u \weaktrans{ } u'
          \transition{\inp{c}{\vec{v}}} u''$ with $(t, u') \in
          \mathcal{R}$ and $(t', u'') \in \mathcal{R}$.
	    \item If $s \in \ptates$ then $\mu(t, D) = \mu(u, D)$ for all classes $D \in \mathcal{T}/\mathcal{R}$.
    \end{enumerate}
\end{Definition}
This relation follows the standard definition of branching
bisimulation \cite{Glabbeek1996} with additional conditions for
probabilistic configurations and matching quantum information. In
condition II we require that the distinct set of values $V$ must match
and although the names ($\vec{X}_1$ and $\vec{X}_2$) need not be
identical which is either the qubit names ($\vec{q}_1$ and $\vec{q}_2$) or number state names ($\vec{s}_1$ and $\vec{s}_2$), their respective reduced density matrices
($\rho^{\vec{X}_1}(t)$ and $\rho^{\vec{X}_2}(u')$) must.
Condition IV provides the matching on probabilistic configurations
following the approach of \cite{Trcka2008}. 
It is necessary to include the latter condition to ensure that the
probabilities are paired with their respective configurations. This leads to the following definitions. The essential definitions are presented in this paper and the others are provided in the appendix.

\begin{Definition}[Probabilistic Branching Bisimilarity]
Configurations $t$ and $u$ are \emph{probabilistic branching bisimilar}, denoted $t \pbsim u$, if there exists a probabilistic branching bisimulation $\mathcal{R}$ such that $(t,u) \in \mathcal{R}$.
\end{Definition}

\begin{Definition}[Probabilistic Branching Bisimilarity of Processes]
Processes $P$ and $Q$ are \emph{probabilistic branching bisimilar}, denoted $P \pbsim Q$, if and only if for all $\sigma$, $(\vec{x}:\vec{T};\sigma; \emptyset; P) \pbsim (\vec{x}:\vec{T};\sigma; \emptyset; Q)$.
\end{Definition}

\begin{Definition}[Full probabilistic branching bisimilarity]
  Processes $P$ and $Q$ are \emph{full probabilistic branching
    bisimilar}, denoted $P \fpbsim Q$, if for all substitutions $\kappa$
and all quantum states
  $\sigma$, $(\vec{x}:\vec{T};\sigma; \vec{q},\vec{s}; P\kappa) \pbsim (\vec{x}:\vec{T};\sigma; \vec{q},\vec{s}; Q\kappa)$. 
\end{Definition}

In order to state the \emph{congruence} theorem, we need an assumption that processes are typable. Its essential idea is to associate each qubit or number state with a unique owning component of the process. In particular this means that when we consider a process P in a context, $C[P]$, the context cannot manipulate quantum state that is owned by P. The full type system is a straightforward extension of the system from CQP, taking account of number states.
\begin{Theorem}[Full probabilistic branching bisimilarity is a congruence]
  \label{thm:congruence}
  If $P \fpbsim Q$ then for any context $\ctxt{C}{}$, if $\ctxt{C}{P}$ and $\ctxt{C}{Q}$ are typable then $\ctxt{C}{P} \fpbsim \ctxt{C}{Q}$. 
\end{Theorem}

\subsection{Correctness of $Model_{1}$}

We now sketch the proof that $\pmodel \fpbsim \pspec$, which by
Theorem~\ref{thm:congruence} implies that the LOQC CNOT gate
works in any context. 
\begin{proposition}
  $\pmodel \fpbsim \pspec$.
\end{proposition}
 \begin{proof}
 First we prove that $\pmodel \pbsim \pspec$,
   by defining an equivalence relation $\mathcal{R}$ that contains the pair
   $(\newcnfig{\vec{x}:\vec{T}}{\sigma}{\emptyset}{\pmodel},\newcnfig{\vec{x}:\vec{T}}{\sigma}{\emptyset}{\pspec})$
   for all $\sigma$ and is closed under their transitions. 
   $\mathcal{R}$ is defined by taking its equivalence classes to be the
   $F_i(\sigma)$ defined below, for all states $\sigma$, which
   group configurations according to the sequences of observable
   transitions leading to them.
 \[
 \begin{array}{rcl}
 \pname{F_{1}(\sigma,q_{1})} & = & \{f \mid
 (\vec{x}:\vec{T};\sigma;\emptyset;P)\weaktrans{\inp{a}{q_{1}}}f ~\mbox{and}~ P \in E\}\\
 \pname{F_{2}(\sigma,q_{1},q_{2})} & = & \{f \mid
 (\vec{x}:\vec{T};\sigma;\emptyset;P)\weaktrans{\inp{a}{q_{1}}}\weaktrans{\inp{b}{q_{2}}}f ~\mbox{and}~ P \in E\}\\
 \pname{F_{3}(\sigma,q_{2})} & = & \{f  \mid  (\vec{x}:\vec{T};\sigma;\emptyset;P)
 \weaktrans{\inp{a}{q_{1}}} \weaktrans{\inp{b}{q_{2}}}\weaktrans{\outp{out_1}{c_{1}}}f ~\mbox{and}~ P \in E\}\\
  \pname{F_{4}(\sigma)} & = & \{f  \mid  (\vec{x}:\vec{T};\sigma;\emptyset;P)
 \weaktrans{\inp{a}{q_{1}}} \weaktrans{\inp{b}{q_{2}}}\weaktrans{\outp{out_1}{c_{1}}}\weaktrans{\outp{out_2}{c_{2}}}f ~\mbox{and}~ P \in E\}\\
   \pname{F_{5}(\sigma)} & = & \{f  \mid  (\vec{x}:\vec{T};\sigma;\emptyset;P)
 \weaktrans{\inp{a}{q_{1}}} \weaktrans{\inp{b}{q_{2}}}\weaktrans{\outp{out_1}{c_{1}}}\weaktrans{\outp{out_2}{c_{2}}}\weaktrans{\outp{cnt}{y}}f ~\mbox{and}~ P \in E\}\\
 \end{array}
 \]
Here $E$ is $\{\pmodel, \pspec \}$ and we now prove that $\mathcal{R}$ is a probabilistic branching
 bisimulation. It suffices to consider transitions between $F_i$
 classes, as transitions within classes must be $\tau$ and are matched
 by $\tau$.  If $f, g \in F_{1}(\sigma)$ and $f\transition{\inp{a}{q_{1}}}f'$ then $f'\in F_{2}(\sigma)$ and we find $g',g''$ such that $g\weaktrans{}g'\transition{\inp{a}{q_{1}}}g''$ with $g'\in F_{1}(\sigma)$ and $g''\in F_{2}(\sigma)$, so $(f,g')\in\mathcal{R}$ and $(f',g'')\in\mathcal{R}$ as required. 
Transitions from $F_{2}(\sigma)$,$F_{3}(\sigma)$ and $F_{4}(\sigma)$ are matched similarly. There are no transitions from $F_{5}(\sigma)$.
There is no need for a probability calculation (case IV of Definition~\ref{def:pbb}) because the probabilistic configurations do not arise as the measurement results are communicated internally. Finally, because $\pmodel$ and $\pspec$ have no free variables, their equivalence is trivially preserved by substitutions.\qed
\end{proof}

\section{LOQC CNOT Gate: A Second Model}
\label{sec-model}
\label{sec:model}
The first model includes an explicit implementation of the \emph{post-selection} procedure, meaning that the specification process has to include the success probability of $\frac{1}{9}$. We now consider a more abstract model, by introducing a new measurement operator which includes \emph{post-selection} and restricts attention to the successful outcomes. This is achieved by replacing the process $\pmmt$ of our first model by the process $\ppsm$ which performs \emph{post-selective} measurement and enables a simpler specification to be used. The CQP definition of $\pmodelt$ is given as
$\pname{Model_2}(\vec{A}) 
=  (\new \vec{B})(\pname{PolSe_{CT}}(\vec{C})\parallel\pname{CNOT}(\vec{D})\parallel\pname{PSM}(\vec{E}))$.
Processes $\ppolsect$ and $\pcnot$ are defined in the previous model. The process $\ppsm$ is defined as $
\pname{PSM}(\vec{E}) =  \pname{PDet_1}(\vec{F})\parallel\pname{PDet_2}(\vec{G}).
$
We prove that $\pmodelt$ is equivalent to $\pspect$:
\[
\begin{array}{rcl}
\pname{OPCNOT}(\vec{C}) 
=  \inp{c}{s_{0}}\sep\inp{d}{s_{1}}\sep\inp{e}{s_{2}}\sep\inp{f}{s_{3}}\sep\action{s_{2},s_{3}\trans\qgate{H}}\sep\\\action{(s_{0},s_{1}),(s_{2},s_{3})\trans\qgate{CZ}}\sep\action{s_{2},s_{3}\trans\qgate{H}}\sep\outp{h}{s_{0}}\sep\outp{i}{s_{1}}\sep\outp{j}{s_{2}}\sep\outp{k}{s_{3}}\sep\nil\\
\pname{Output}(\vec{D}) 
= \inp{h}{s_{0}}\sep\inp{i}{s_{1}}\sep\inp{j}{s_{2}}\sep\inp{k}{s_{3}}\sep\outp{out_{1}}{\measure s_{1}}\sep\outp{out_{2}}{\measure s_{3}}\sep\nil\\
\pname{Specification_2}(\vec{A})
=  (\new \vec{E})(\pname{PolSe_{CT}}(\vec{B})\parallel\pname{OPCNOT}(\vec{C})\parallel\pname{Output}(\vec{D}))
\end{array}
\]
The analysis of $\pmodelt$ and the proof of its correctness are provided in the Appendix.

\section{Conclusion and Future Work}
\label{sec-conclusion}
\label{sec:conclusion}
The main contribution of this paper is the extension of theory of equivalence of CQP to verify linear optical quantum computing. This is the first work in using quantum process calculus to verify a physical realisation of quantum computing. We have defined the linear optical elements in CQP, and have described and analysed two models of the linear optical experimental system that demonstrates a CNOT gate. Using our second model, we have also described and verified post-selection in CQP.

These two models use different measurement semantics in order to work at different levels of abstraction. This shows that the process calculus is flexible enough to support a range of descriptions, from detailed hardware implementations up to more abstract specifications. The importance of process calculus is that it provides a systematic methodology for verification of quantum systems. The essential property that the equivalence is a congruence guarantees that equivalent processes remain equivalent in any context, and supports equational reasoning. The fact that CQP can also express classical behaviour means that we have a uniform framework in which to analyze classical and quantum computation and communication. 

Shor's algorithm operating on four qubits using the basic linear optical elements has been demonstrated \cite{Brien2009a}. In this paper, we present the modelling of these elements with a future aim to formally analyse quantum algorithms in CQP using LOQC. 
This provides a platform to learn about quantum complexity in LOQC using CQP and also to verify it.
The long-term goal is to develop software for automated analysis of CQP models,  following the established work in classical process calculus and recent work on automated equivalence checking of concurrent quantum programs \cite{Ebrahim2014}.


\bibliographystyle{eptcs}

\bibliography{report_main}

\begin{thebibliography}{10}
\providecommand{\bibitemdeclare}[2]{}
\providecommand{\surnamestart}{}
\providecommand{\surnameend}{}
\providecommand{\urlprefix}{Available at }
\providecommand{\url}[1]{\texttt{#1}}
\providecommand{\href}[2]{\texttt{#2}}
\providecommand{\urlalt}[2]{\href{#1}{#2}}
\providecommand{\doi}[1]{doi:\urlalt{http://dx.doi.org/#1}{#1}}
\providecommand{\bibinfo}[2]{#2}

\bibitemdeclare{inproceedings}{Ebrahim2014}
\bibitem{Ebrahim2014}
\bibinfo{author}{E.~\surnamestart Ardeshir-Larijani\surnameend},
  \bibinfo{author}{S.~J. \surnamestart Gay\surnameend} \&
  \bibinfo{author}{R.~\surnamestart Nagarajan\surnameend}
  (\bibinfo{year}{2014}): \emph{\bibinfo{title}{Verification of Concurrent
  Quantum Protocols by Equivalence Checking}}.
\newblock In: {\sl \bibinfo{booktitle}{Proceedings of the 20th International
  Conference on Tools and Algorithms for the Construction and Analysis of
  Systems (TACAS)}}, \bibinfo{volume}{8413}, \bibinfo{publisher}{LNCS}, pp.
  \bibinfo{pages}{500--514}, \doi{10.1007/978-3-642-54862-8_42}.

\bibitemdeclare{inproceedings}{Davidson2011}
\bibitem{Davidson2011}
\bibinfo{author}{T.~A.~S. \surnamestart Davidson\surnameend},
  \bibinfo{author}{S.~J. \surnamestart Gay\surnameend},
  \bibinfo{author}{R.~\surnamestart Nagarajan\surnameend} \&
  \bibinfo{author}{I.~V. \surnamestart Puthoor\surnameend}
  (\bibinfo{year}{2011}): \emph{\bibinfo{title}{Analysis of a Quantum Error
  Correcting Code using Quantum Process Calculus}}.
\newblock In: {\sl \bibinfo{booktitle}{Proceedings of the International
  Workshop on Quantum Physics and Logic (QPL)}}, \bibinfo{volume}{95},
  \bibinfo{publisher}{EPTCS}, pp. \bibinfo{pages}{67--80},
  \doi{10.4204/EPTCS.95.7}.

\bibitemdeclare{phdthesis}{DavidsonThesis}
\bibitem{DavidsonThesis}
\bibinfo{author}{Timothy A.~S. \surnamestart Davidson\surnameend}
  (\bibinfo{year}{2011}): \emph{\bibinfo{title}{Formal Verification Techniques
  using Quantum Process Calculus}}.
\newblock Ph.D. thesis, \bibinfo{school}{University of Warwick}.

\bibitemdeclare{inproceedings}{Feng2011}
\bibitem{Feng2011}
\bibinfo{author}{Yuan \surnamestart Feng\surnameend}, \bibinfo{author}{Runyao
  \surnamestart Duan\surnameend} \& \bibinfo{author}{Mingsheng \surnamestart
  Ying\surnameend} (\bibinfo{year}{2011}): \emph{\bibinfo{title}{Bisimulation
  for quantum processes}}.
\newblock In: {\sl \bibinfo{booktitle}{Proceedings of the 38th Annual ACM
  Symposium on Principles of Programming Languages}}, \bibinfo{publisher}{ACM},
  pp. \bibinfo{pages}{523--534}, \doi{10.1145/1926385.1926446}.

\bibitemdeclare{inproceedings}{Arnold2013}
\bibitem{Arnold2013}
\bibinfo{author}{S.~\surnamestart Franke-Arnold\surnameend},
  \bibinfo{author}{S.~J. \surnamestart Gay\surnameend} \&
  \bibinfo{author}{I.~V. \surnamestart Puthoor\surnameend}
  (\bibinfo{year}{2013}): \emph{\bibinfo{title}{Quantum process calculus for
  linear optical computing}}.
\newblock In: {\sl \bibinfo{booktitle}{Proceedings of the 5th Conference on
  Reversible Computation (RC)}}, \bibinfo{volume}{7948},
  \bibinfo{publisher}{LNCS}, pp. \bibinfo{pages}{234--246},
  \doi{10.1007/978-3-642-38986-3_19}.

\bibitemdeclare{inproceedings}{Gay2005}
\bibitem{Gay2005}
\bibinfo{author}{Simon~J. \surnamestart Gay\surnameend} \&
  \bibinfo{author}{Rajagopal \surnamestart Nagarajan\surnameend}
  (\bibinfo{year}{2005}): \emph{\bibinfo{title}{Communicating {Q}uantum
  {P}rocesses}}.
\newblock In: {\sl \bibinfo{booktitle}{Proceedings of the 32nd Annual ACM
  Symposium on Principles of Programming Languages}}, \bibinfo{publisher}{ACM},
  pp. \bibinfo{pages}{145--157}, \doi{10.1145/1040305.1040318}.

\bibitemdeclare{article}{Gay2006a}
\bibitem{Gay2006a}
\bibinfo{author}{Simon~J. \surnamestart Gay\surnameend} \&
  \bibinfo{author}{Rajagopal \surnamestart Nagarajan\surnameend}
  (\bibinfo{year}{2006}): \emph{\bibinfo{title}{{Types and Typechecking for
  Communicating Quantum Processes}}}.
\newblock {\sl \bibinfo{journal}{Mathematical Structures in Computer Science}}
  \bibinfo{volume}{16}(\bibinfo{number}{3}), pp. \bibinfo{pages}{375--406},
  \doi{10.1017/S0960129506005263}.

\bibitemdeclare{article}{Glabbeek1996}
\bibitem{Glabbeek1996}
\bibinfo{author}{Rob~J. \surnamestart van Glabbeek\surnameend} \&
  \bibinfo{author}{W.~Peter \surnamestart Weijland\surnameend}
  (\bibinfo{year}{1996}): \emph{\bibinfo{title}{Branching time and abstraction
  in bisimulation semantics}}.
\newblock {\sl \bibinfo{journal}{Journal of the ACM}}
  \bibinfo{volume}{43}(\bibinfo{number}{3}), pp. \bibinfo{pages}{555--600},
  \doi{10.1145/233551.233556}.

\bibitemdeclare{misc}{IDQ2001a}
\bibitem{IDQ2001a}
\bibinfo{author}{\surnamestart IDQ\surnameend}: \emph{\bibinfo{title}{ID
  Quantique}}.
\newblock \urlprefix\url{http://www.idquantique.com/company/presentation.html}.

\bibitemdeclare{article}{Knill2001}
\bibitem{Knill2001}
\bibinfo{author}{E.~\surnamestart Knill\surnameend},
  \bibinfo{author}{R.~\surnamestart Laflamme\surnameend} \&
  \bibinfo{author}{G.~J. \surnamestart Milburn\surnameend}
  (\bibinfo{year}{2001}): \emph{\bibinfo{title}{A scheme for efficient quantum
  computation with linear optics}}.
\newblock {\sl \bibinfo{journal}{Nature}} \bibinfo{volume}{409},
  p.~\bibinfo{pages}{46}, \doi{10.1038/35051009}.

\bibitemdeclare{inproceedings}{Kubota2012}
\bibitem{Kubota2012}
\bibinfo{author}{T.~\surnamestart Kubota\surnameend},
  \bibinfo{author}{Y.~\surnamestart Kakutani\surnameend},
  \bibinfo{author}{G.~\surnamestart Kato\surnameend},
  \bibinfo{author}{Y.~\surnamestart Kawano\surnameend} \&
  \bibinfo{author}{H.~\surnamestart Sakurada\surnameend}
  (\bibinfo{year}{2012}): \emph{\bibinfo{title}{Application of a process
  calculus to security proofs of quantum protocols}}.
\newblock In: {\sl \bibinfo{booktitle}{Proceedings of WORLDCOMP/FCS2012}}.

\bibitemdeclare{book}{Milner1989}
\bibitem{Milner1989}
\bibinfo{author}{Robin \surnamestart Milner\surnameend} (\bibinfo{year}{1989}):
  \emph{\bibinfo{title}{Communication and {C}oncurrency}}.
\newblock \bibinfo{publisher}{Prentice-Hall}.

\bibitemdeclare{book}{Milner1999}
\bibitem{Milner1999}
\bibinfo{author}{Robin \surnamestart Milner\surnameend} (\bibinfo{year}{1999}):
  \emph{\bibinfo{title}{Communicating and Mobile Systems: the Pi-Calculus}}.
\newblock \bibinfo{publisher}{Cambridge University Press}.

\bibitemdeclare{article}{Myers2005}
\bibitem{Myers2005}
\bibinfo{author}{C.~R. \surnamestart Myers\surnameend} \&
  \bibinfo{author}{R.~\surnamestart Laflamme\surnameend}
  (\bibinfo{year}{2005}): \emph{\bibinfo{title}{Linear Optics Quantum
  Computation: an Overview}}.
\newblock {\sl \bibinfo{journal}{arXiv: quant-ph/0512104v1}}.

\bibitemdeclare{book}{Nielsen2000}
\bibitem{Nielsen2000}
\bibinfo{author}{M.~A. \surnamestart Nielsen\surnameend} \&
  \bibinfo{author}{I.~L. \surnamestart Chuang\surnameend}
  (\bibinfo{year}{2000}): \emph{\bibinfo{title}{Quantum Computation and Quantum
  Information}}.
\newblock \bibinfo{publisher}{Cambridge University Press}.

\bibitemdeclare{article}{Brien2003}
\bibitem{Brien2003}
\bibinfo{author}{J.~L. \surnamestart O'Brien\surnameend},
  \bibinfo{author}{G.~J. \surnamestart Pryde\surnameend},
  \bibinfo{author}{A.~G. \surnamestart White\surnameend},
  \bibinfo{author}{T.~C. \surnamestart Ralph\surnameend} \&
  \bibinfo{author}{D.~\surnamestart Branning\surnameend}
  (\bibinfo{year}{2003}): \emph{\bibinfo{title}{Demonstration of an all-optical
  quantum controlled-NOT gate}}.
\newblock {\sl \bibinfo{journal}{Nature}} \bibinfo{volume}{426}, p.
  \bibinfo{pages}{264}, \doi{10.1038/nature02054}.

\bibitemdeclare{article}{Brien2009a}
\bibitem{Brien2009a}
\bibinfo{author}{A.~\surnamestart Politi\surnameend}, \bibinfo{author}{J.~C.~F.
  \surnamestart Matthews\surnameend} \& \bibinfo{author}{J.~L. \surnamestart
  O'Brien\surnameend} (\bibinfo{year}{2009}): \emph{\bibinfo{title}{Shor's
  Quantum Factoring Algorithm on a Photonic Chip}}.
\newblock {\sl \bibinfo{journal}{Science}} \bibinfo{volume}{325}, p.
  \bibinfo{pages}{1221}, \doi{10.1126/science.1173731}.

\bibitemdeclare{article}{Ralph2002}
\bibitem{Ralph2002}
\bibinfo{author}{T.~C. \surnamestart Ralph\surnameend}, \bibinfo{author}{N.~K.
  \surnamestart Lanford\surnameend}, \bibinfo{author}{T.~B. \surnamestart
  Bell\surnameend} \& \bibinfo{author}{A.~G. \surnamestart White\surnameend}
  (\bibinfo{year}{2002}): \emph{\bibinfo{title}{Linear optical controlled-NOT
  gate in the coincidence basis}}.
\newblock {\sl \bibinfo{journal}{Physical Review Letters A}}
  \bibinfo{volume}{65}, pp. \bibinfo{pages}{062324--1},
  \doi{10.1103/PhysRevA.65.062324}.

\bibitemdeclare{article}{Trcka2008}
\bibitem{Trcka2008}
\bibinfo{author}{Nikola \surnamestart Tr\v{c}ka\surnameend} \&
  \bibinfo{author}{Sonja \surnamestart Georgievska\surnameend}
  (\bibinfo{year}{2008}): \emph{\bibinfo{title}{Branching Bisimulation
  Congruence for Probabilistic Systems}}.
\newblock {\sl \bibinfo{journal}{Electronic Notes in Theoretical Computer
  Science}} \bibinfo{volume}{220}(\bibinfo{number}{3}), pp. \bibinfo{pages}{129
  -- 143}, \doi{10.1016/j.entcs.2008.11.023}.

\bibitemdeclare{article}{Wright1994}
\bibitem{Wright1994}
\bibinfo{author}{Andrew~K. \surnamestart Wright\surnameend} \&
  \bibinfo{author}{Matthias \surnamestart Felleisen\surnameend}
  (\bibinfo{year}{1994}): \emph{\bibinfo{title}{A syntactic approach to type
  soundness}}.
\newblock {\sl \bibinfo{journal}{Information and Computation}}
  \bibinfo{volume}{115}(\bibinfo{number}{1}), pp. \bibinfo{pages}{38--94},
  \doi{10.1006/inco.1994.1093}.

\end{thebibliography}

\newpage

\section{Appendix}
\label{sec:app}
\label{sec-app}
\subsection{Definitions and Lemmas for Equivalence}
 \begin{Definition}[Context]
   \label{def:cong_context}
   A \emph{context} $C$ is a process with a non-degenerate occurrence of $\nil$ replaced by a hole, $[\cdot]$. Formally,
   \[ C \bnf [ ] \mid (C \parcomp P) \alt \alpha.C + P \alt \alpha.C  \alt (\new x\chant{T})C \] 
   for $\alpha \in \{\inp{e}{\tilde{x}:\tilde{T}}, \outp{e}{\tilde{e}}, \{e\}, (\qbit x), (\ns r) \}$.
 \end{Definition}
 \begin{Definition}[Congruence]
   An equivalence relation $\mathcal{R}$ on processes is a
   \emph{congruence} if $(C[P], C[Q]) \in \mathcal{R}$ whenever $(P, Q)
   \in \mathcal{R}$ and $C$ is a context. 
 \end{Definition}
 \begin{Definition}[Non-input, non-qubit or non-number state context]
   A \emph{non-input, non-qubit or non-number state context} is a context in which the hole
   does not appear under an input or qubit  and number state declaration.
 \end{Definition}
 \begin{Definition}[Non-input, non-qubit or non-number state congruence]
   An equivalence relation $\mathcal{R}$ on processes is a
   \emph{non-input, non-qubit or non-number state congruence} if $(C[P], C[Q]) \in
   \mathcal{R}$ whenever $(P, Q) \in \mathcal{R}$ and $C$ is a
   non-input, non-qubit or non-number state context. 
 \end{Definition}
 The first lemma provides a general form for representing mixed configurations related by internal transitions. Due to space constraints the proofs of all lemmas and theorems are not provided in this paper.
 \begin{Lemma}[General form of internal transitions]
  If $t = \Dist{\substack{ab \in I_{kl}\\kl \in J}}{g_{abkl}}(\vec{x}:\vec{T};\sigma_{abkl};\vec{q},\vec{s};\ltrm{\vec{y}\vec{z}}{P}{\vec{w}_{abkl}})$ and $t\weaktrans{}t'$ then there exist sets $I'_{kl}$ such that $t' = \Dist{\substack{ab \in I'_{kl}\\kl \in J}}{g'_{abkl}}(\vec{x}:\vec{T};\sigma'_{abkl};\vec{q}',\vec{s}';\ltrm{\vec{y}'\vec{z}'}{P'}{\vec{w}'_{abkl}})$.
  \end{Lemma}  
  The following $3$ lemmas prove that the state of qubits and number states that are not owned by a particular process is unaffected by any transitions of that process.
 \begin{Lemma}[External state independence for $\transitionv{}$]
    \label{lem:lemma2}
  If $\Gamma$; $\vec{s} \vdash e : T$ and $t\transitionv{}t'$ where $t = (\vec{s}:\vec{\NS}, \vec{q}:\vec{\Qbit},\vec{r}:\vec{\Qbit};\qstore{\vec{s}\vec{q}\vec{r}}{\ket{\psi}};\vec{q},\vec{s};e)$ then $\rho^{\vec{q}\vec{r}}(t) = \rho^{\vec{q}\vec{r}}(t')$
  \end{Lemma}
    \begin{Lemma}[External state independence for $\transitione{}$]
    \label{lem:lemma3}
 If $\Gamma$; $\vec{s} \vdash e : T$ and $t\transitione{}t'$ where $t = \Dist{kl \in I} g_{kl}(\vec{s}:\vec{\NS}, \vec{q}:\vec{\Qbit},\vec{r}:\vec{\Qbit};\qstore{\vec{s}\vec{q}\vec{r}}{\ket{\psi_{kl}}};\vec{q},\vec{s};\ltrm{\vec{y}}{e}{\vec{w}_{kl}})$ then $\rho^{\vec{q}\vec{r}}(t) = \rho^{\vec{q}\vec{r}}(t')$
  \end{Lemma}
  \begin{Lemma}[External state independence for $\transition{\tau}$]
    \label{lem:lemma4}
 If $\Gamma$; $\vec{s} \vdash P$ and $t\transition{\tau}t'$ where $t = \Dist{kl \in I} g_{kl}(\vec{s}:\vec{\NS}, \vec{q}:\vec{\Qbit},\vec{r}:\vec{\Qbit};\qstore{\vec{s}\vec{q}\vec{r}}{\ket{\psi_{kl}}};\vec{q},\vec{s};\ltrm{\vec{y}}{P}{\vec{w}_{kl}})$ then $\rho^{\vec{q}\vec{r}}(t) = \rho^{\vec{q}\vec{r}}(t')$
  \end{Lemma}
The next lemma proves that the action of a context on the quantum state is independent of the quantum subsystem owned by a process.
  \begin{Lemma}[Independence of context transitions]
  \label{lem:lemma4b}
  Assume that $\Gamma$; $\vec{s}_{R} \vdash R$. Let $t$ and $u$ be configurations where
\[
t = \Dist{kl \in I} g_{kl}(\vec{x}:\vec{T};\qstore{\vec{q}_{P}\vec{q}_{R}\vec{q}_{E}\vec{s}_{P}\vec{s}_{R}\vec{s}_{E}}{\ket{\psi_{kl}}};\vec{q}_{P},\vec{q}_{R},\vec{s}_{P},\vec{s}_{R};\ltrm{\vec{y}}{R}{\vec{w}_{R}})
\]
\[
u = \Dist{mn \in J} h_{mn}(\vec{x}:\vec{T};\qstore{\vec{q}_{Q}\vec{q}_{R}\vec{q}_{E}\vec{s}_{Q}\vec{s}_{R}\vec{s}_{E}}{\ket{\phi_{mn}}};\vec{q}_{Q},\vec{q}_{R},\vec{s}_{Q},\vec{s}_{R};\ltrm{\vec{y}}{R}{\vec{w}_{R}})
\]
If $\rho^{\vec{q}_{P}\vec{q}_{E}\vec{s}_{P}\vec{s}_{E}}(t) = \rho^{\vec{q}_{Q}\vec{q}_{E}\vec{s}_{Q}\vec{s}_{E}}(u)$ and $t \transition{\tau}t'$ where $t = \Dist{\substack{kl \in I'_{ab}\\ab \in K}} g'_{klab}(\vec{x}:\vec{T};\qstore{\vec{q}_{P}\vec{q}'_{R}\vec{q}_{E}\vec{s}_{P}\vec{s}'_{R}\vec{s}_{E}}{\ket{\psi_{klab}}};\omega_{P},\omega'_{R};\ltrm{\vec{y}'}{R'}{\vec{w}_{R_{ab}}})$ then there exists $u = \Dist{\substack{mn \in J'_{ab}\\ab \in K}} h'_{mnab}(\vec{x}:\vec{T};\qstore{\vec{q}_{Q}\vec{q}'_{R}\vec{q}_{E}\vec{s}_{Q}\vec{s}'_{R}\vec{s}_{E}}{\ket{\phi_{mnab}}};\omega_{Q},\omega'_{R};\ltrm{\vec{y}'}{R'}{\vec{w}_{R_{ab}}})$ such that $u \transition{\tau} u'$ and $\rho^{\vec{q}_{P}\vec{q}_{E}\vec{s}_{P}\vec{s}_{E}}(t') = \rho^{\vec{q}_{Q}\vec{q}_{E}\vec{s}_{Q}\vec{s}_{E}}(u')$
\end{Lemma}
The next two lemmas prove some simple results which are used in the proof of Theorem~\ref{thm:cong_parallel_preservation_configurations}.
\begin{Lemma}
  \label{lem:lemma5}
Let $t = \Dist{kl \in I} g_{kl}t_{kl}$ and $t' = \Dist{kl \in I} g_{kl}t'_{kl}$ then $t \transition{\alpha} t'$ if and only if $\forall_{kl \in I}(t_{kl} \transition{\alpha} t'_{kl})$ for $ \alpha \in \{\inp{.}\cdot,\tau\}$
\end{Lemma}
\begin{Lemma}
    \label{lem:lemma6}
Let $t_{mn} = \Dist{kl \in I_{mn}} g_{klmn}(\vec{x}:\vec{T};\sigma_{klmn};\omega;\ltrm{\vec{y}}{P}{\vec{w}_{klmn}})$ and $t_{klmn} = (\vec{x}:\vec{T};\sigma_{klmn};\omega;P\{\vec{w}_{klmn}/\vec{y}\})$ then $\forall_{mn \in J, kl \in I_{mn}}.(t_{klmn}\transitionp{\inp{c}{\vec{u}_{mn},\vec{q},\vec{s}}} t'_{klmn})$ if and only if $\forall_{mn \in J}.(t_{mn} \transitionp{\inp{c}{\vec{u}_{mn},\vec{q},\vec{s}}}  t'_{mn})$
  \end{Lemma}
We are now in a position to prove that bisimilarity is preserved by parallel composition. To prove this, we define an equivalence relation that contains the pair $((\vec{x}:\vec{T};\sigma;\emptyset;P\parallel R),(\vec{x}:\vec{T};\sigma;\emptyset;Q\parallel R))$ and that is closed under transitions from these configurations.
 \begin{Theorem}[Parallel preservation for configurations]
   \label{thm:cong_parallel_preservation_configurations}
   Assume that $\ptyped{\Gamma}{P}$, $\ptyped{\Gamma}{Q}$,
   $\ptyped{\Gamma}{P \parallel R}$, and $\ptyped{\Gamma}{Q \parallel
     R}$. If $(\vec{x}:\vec{T};\sigma; \emptyset; P) \pbsim (\vec{x}:\vec{T};\sigma; \emptyset; Q)$
   then $(\vec{x}:\vec{T};\sigma; \emptyset; P \parallel R) \pbsim (\vec{x}:\vec{T};\sigma; \emptyset; Q\parallel R)$. 
 \end{Theorem}
  Using this result, we prove that the bisimilarity of processes is preserved by parallel composition.
 \begin{Theorem}[Parallel Preservation]
   \label{thm:cong_parallel_preservation}
   If $P \pbsim Q$ then for any process $R$ such that
   $\ptyped{\Gamma}{P \parallel R}$ and $\ptyped{\Gamma}{Q \parallel R}$ then $P \parallel R \pbsim Q \parallel R$.
 \end{Theorem}
 We now consider preservation with respect to other process constructions and can be shown that probabilistic branching bisimilarity is preserved by all process constructs except input and qubit or number state declarations. 
 \begin{Lemma}
 \label{lem:lemma7}
Probabilistic branching bisimilarity is preserved by output prefix, action prefix, channel restriction and non-deterministic choice.
 \end{Lemma}
 \begin{Theorem}[Probabilistic branching bisimilarity is a non-input congruence]
   \label{thm:cong_noninput_congruence}
   If $P \pbsim Q$ and for any non-input, non-qubit or non-number state context $C$ if $\ptyped{\Gamma}{C[P]}$ and $\ptyped{\Gamma}{C[Q]}$ then $C[P] \pbsim C[Q]$.
 \end{Theorem}
  
 \subsection{Execution of $Model_{2}$:}
Let $t = (\emptyset;\emptyset;\emptyset;\pmodelt)$ be the initial configuration. Like in previous case after receiving input qubits, we get the configuration as,
$(q_{1} : \Qbit, q_{2} : \Qbit, q_{1}q_{2} = \ket{\phi}_{q};q_{1},q_{2};(\ppolsect ' \parallel\pcnot\parallel\ppsm))$.
As before the qubits are converted to the number states after some $\tau$ operations and the configuration is now,
 \[
 \begin{array}{l}
(\vec{s} : \vec{\NS};\vec{s} = \ket{\phi}_{s};s_{0},s_{1},s_{2},s_{3};(\ppolsect '' \parallel\pcnot\parallel\ppsm))
\end{array}
\]
After another set of $\tau$ transitions corresponding to the $\pcnot$ process, we get the state $\ket{\phi}_{out}$ which is given by Eq.~\ref{Eq3}. The configuration now becomes
$(\vec{s} : \vec{\NS};\vec{s} = \ket{\phi}_{out};s_{0},s_{1},s_{2},s_{3};(\pcnot' \parallel\ppsm))$.
The execution of $Model_{2}$ is similar to that of $Model_{1}$ and differs only in the measurement. Here the detectors perform a \emph{post-selective} measurement giving rise to the following mixed configuration:
 \[
 \begin{array}{l}
\Dist{\substack{ij \in \{0,1\}, i \neq j\\kl \in \{0,1\}, k \neq l}}{g_{ij}h_{ijkl}}(\vec{s} : \vec{\NS};\vec{s} = \ket{\phi_{ijkl}};s_{0},s_{1},s_{2},s_{3};\ltrm{\vec{y}}{\ppsm'}j,l)
\end{array}
\]
The \emph{post-selective} measurement outcomes ($\vec{y}$) are then given as output to the environment resulting in a probabilistic configuration given as
$\Prob{ij \in \{0,1\},kl \in \{0,1\}}g_{ij}h_{ijkl}(\vec{s} : \vec{\NS};\vec{s} = \ket{\phi_{ijkl}};s_{0},s_{1},s_{2},s_{3};\ltrm{\vec{y}}{\nil}j,l)$.

Another significant difference between the models is in the communication of the measurement outcomes. In $Model_{1}$, the outcomes were communicated internally and hence did not give a probabilistic configuration, which is not the case for $Model_{2}$. 

\subsection{Correctness of $Model_{2}$}

\begin{proposition}
  $\pmodelt \fpbsim \pspect$.
\end{proposition}
 \begin{proof}
 We have similar equivalence classes as in the previous case:
  \[
 \begin{array}{rcl}
 \pname{F_{1}(\sigma,q_{1})} & = & \{f \mid
 (\vec{x}:\vec{T};\sigma;\emptyset;P)\weaktrans{\inp{a}{q_{1}}}f ~\mbox{and}~ P \in E\}\\
 \pname{F_{2}(\sigma,q_{1},q_{2})} & = & \{f \mid
 (\vec{x}:\vec{T};\sigma;\emptyset;P)\weaktrans{\inp{a}{q_{1}}}\weaktrans{\inp{b}{q_{2}}}f ~\mbox{and}~ P \in E\}\\
 \pname{F_{3}(\sigma,q_{2})} & = & \{f  \mid  (\vec{x}:\vec{T};\sigma;\emptyset;P)
 \weaktrans{\inp{a}{q_{1}}} \weaktrans{\inp{b}{q_{2}}}\weaktrans{\outp{out_1}{c_{1}}}f ~\mbox{and}~ P \in E\}\\
  \pname{F_{4}(\sigma)} & = & \{f  \mid  (\vec{x}:\vec{T};\sigma;\emptyset;P)
 \weaktrans{\inp{a}{q_{1}}} \weaktrans{\inp{b}{q_{2}}}\weaktrans{\outp{out_1}{c_{1}}}\weaktrans{\outp{out_2}{c_{2}}}f ~\mbox{and}~ P \in E\}
 \end{array}
 \]
Here $E$ is $\{\pmodelt, \pspect \}$ and the proof is similar to the previous case. In $\pmodelt$, we will always get a correct output since we do not consider any error and the probability of getting one of the outputs is $\frac{1}{4}$. 
Similar to the previous proof, here we have no transitions from $F_4(\sigma)$.\qed
 \end{proof}

\end{document}